\documentclass[11pt]{article}

\usepackage{soul}
\usepackage{etex}
\usepackage{amsmath}
\usepackage{amssymb}
\usepackage{amsthm}
\usepackage{amsfonts}
\usepackage{epsfig}
\usepackage{graphicx}
\usepackage{subcaption}
\usepackage{latexsym}
\usepackage{url}
\usepackage{ifthen}
\usepackage{comment}
\usepackage{enumerate}
\usepackage{multirow}
\usepackage[english]{babel}
\usepackage{texnansi}
\usepackage{textcomp}
\usepackage{color}
\usepackage{calc}
\usepackage[normalem]{ulem}
\usepackage{array}
\usepackage{booktabs}
\usepackage[symbol*]{footmisc}
\setfnsymbol{wiley}
\usepackage{relsize}

\usepackage[margin=10pt,font=small,labelfont=bf]{caption}
\usepackage[hypertexnames=false,hyperfootnotes=false,bookmarks=false]{hyperref}
\usepackage{sectsty}
\usepackage[capitalise]{cleveref}
\usepackage{tikz}

\usepackage{setspace}
\usepackage[margin=1in]{geometry}

\provideboolean{lucidabr}

\ifthenelse{\boolean{lucidabr}}{%
  \usepackage{lucidabr}
}{%
  \usepackage{lmodern}
}

\allsectionsfont{\large\sffamily}
\makeatletter
\def\@seccntformat#1{\csname the#1\endcsname.\quad}
\makeatother

\providecommand{\keywords}[1]{\small{\textit{Keywords:}} #1}
\usepackage{tikz}
\usepackage{pgfpages}
\usepackage{pgfplots}
\usepgfplotslibrary{colormaps}
\usepackage{pgfplotstable}
\usetikzlibrary{arrows,decorations,patterns,trees,matrix,calc,shapes}

\DeclareMathOperator*{\argmax}{\text{argmax}}

\renewcommand\max{\mathop{\text{max}}\limits}
\renewcommand\min{\mathop{\text{min}}\limits}
\renewcommand\lim{\mathop{\text{lim}}\limits}

\tikzstyle{every picture} += [>=stealth]

\tikzset{axis/.style={semithick, line join=miter}}

\pgfdeclarelayer{background}
\pgfdeclarelayer{foreground}
\pgfsetlayers{background,main,foreground}

\pgfplotstableset{%
  font=\small,
  every head row/.style={before row=\toprule[1pt], after row=\midrule},
  every last row/.style={after row=\bottomrule[1pt]}}

\def\coursename#1{%
  \def\ctemp{#1}%
  \ifx\ctemp\@empty
  \def\insertcoursename{}
  \else
  \def\insertcoursename{\ignorespaces#1}
  \fi
}
\coursename{}



\usepackage{natbib}
 \bibpunct[, ]{(}{)}{,}{a}{}{,}%

\newtheorem{theorem}{{\bfseries\sffamily Theorem}}

\newtheorem{corollary}{Corollary}

\newtheorem{definition}{Definition}
\newtheorem{example}{Example}

\newtheorem{lemma}{Lemma}

\newtheorem{proposition}{Proposition}
\newtheorem{remark}{Remark}

\newcommand{\mathbbm}[1]{\text{\usefont{U}{bbm}{m}{n}#1}} %

\usepackage{ragged2e}
\usepackage{enumitem}
\usepackage{titlesec}
\usepackage{mathtools}
\usepackage{tikz}
\usetikzlibrary{positioning, arrows.meta}

\usepackage{color-edits}
\addauthor{Chiwei}{cyan}
\allowdisplaybreaks

\titleformat*{\section}{\sffamily\Large\bfseries}
\titleformat*{\subsection}{\sffamily\large\bfseries}
\titleformat*{\subsubsection}{\sffamily\large\bfseries}

    \makeatletter
    \def\@fnsymbol#1{\ensuremath{\ifcase#1\or *\or \dagger\or \ddagger\or
   \mathsection\or \mathparagraph\or \|\or **\or \dagger\dagger
   \or \ddagger\ddagger \else\@ctrerr\fi}}
    \makeatother

\title{\LARGE\bfseries\sffamily 
{\color{black}Restricting Entries to All-Pay Contests}\thanks{The first two authors contributed equally. An extended abstract of this paper appeared in the \emph{Proceedings of the 25th ACM Conference on Economics and Computation (EC'24)}. We thank the review team for their valuable comments, which have helped us significantly improve the paper.

}
}

\begin{document}
\author{Fupeng Sun\thanks{Imperial Business School, Imperial College London, f.sun23@imperial.ac.uk.}
\and Yanwei Sun\thanks{Imperial Business School, Imperial College London, yanwei@imperial.ac.uk.}
\and Chiwei Yan\thanks{Department of Industrial Engineering and Operations Research,
UC Berkeley, chiwei@berkeley.edu.}
\and Li Jin\thanks{UM-SJTU Joint
Institute, Shanghai Jiao Tong University, li.jin@sjtu.edu.cn.}}

\date{}
\maketitle

\begin{abstract}

We study an all-pay contest in which players with low abilities are filtered out before competing for prizes. %
We consider a setting where the designer admits a certain number of top players. The admitted players update their beliefs based on the signal that their abilities are among the top, which leads to posterior beliefs that, even under i.i.d. priors, are correlated and depend on each player's private ability. 
We find that all effects of this elimination mechanism---including the reduction in the number of admitted players and the resulting updated beliefs---are captured by an \textit{inflated ability}. 
A symmetric and strictly increasing equilibrium strategy exists if and only if this inflated ability is increasing in the player's true ability. Under this condition, we explicitly characterize the unique strictly increasing Bayesian equilibrium strategy.
Focusing on a winner-take-all prize structure, we find that each admitted player's effort strictly decreases as the admitted number increases. As a result, it is optimal to admit only two players in terms of maximizing the expected highest effort.
Finally, in a two-stage extension, we find that there does not exist a symmetric and strictly increasing equilibrium strategy.
\end{abstract}


{\normalsize

\section{Introduction}
\label{sec:intro}

All-pay contests are scenarios where players invest costly effort to compete for valuable prizes. A common practice among contest designers is to restrict entries of players.
One of the main reasons is that the elimination phase can enhance the contest's competitiveness, sometimes motivating contestants to exert greater effort (see, e.g.,  \citealt{Moldovanu_2006_Contest_Architecture_Sub-elimination, brown_ms_2014_eliminationtournament, fu_aermicro_2022_disclosure_sec}).

As a first example, research funding agencies, such as the National Science Foundation in the US, often have very limited resources of reviewers and panelists, and they thus need to filter potential low-quality submissions without going through the costly formal review process. This is often done by asking proposers to submit a brief letter of intent documenting the proposers' names and affiliations, proposal title, and a synopsis of the topic, prior to full proposal submissions. Proposals whose chances of success are very small are discouraged from formal submissions.\footnote{See details at \url{https://www.nsf.gov/pubs/policydocs/pappguide/nsf14001/gpg_1.jsp}.}  %
Funding agencies also want proposers to exert more effort in proposal writings, which helps them
to think about the big picture and put their projects into perspective, ultimately leading to better and socially more valuable research.\footnote{We acknowledge that sometimes, funding agencies might want to identify the best applicant. This goal is not explicitly considered in our model.}

Another example is the RoboMaster competition, an annual intercollegiate robotics contest held in Shenzhen, China. %
The event aims to maximize audience engagement by creating a highly competitive environment that motivates participants to invest significant effort.\footnote{There are two main reasons the organizers seek high efforts from participants. First, their mission is to ``lead global robotics competitions and drive the development of the robotics industry through the most rigorous competition rules'', see details at \url{https://www.robomaster.com/en-US/robo/overview?djifrom=nav}. Achieving this goal demands substantial effort from teams to keep the competition engaging and impactful. Second, the same webpage describes various spin-offs, which rely on intense competition to attract a broader audience and remain commercially viable.
}
To this end, they select only the most promising teams during the registration phase, a process facilitated by screening resumes to assess each team's potential. Moreover, by disclosing the number of total as well as admitted teams, the designers provide participants with a gauge of the contest's competitiveness.

\smallskip

Despite the prevalent use of entry restrictions, 
the impact of such elimination on players' efforts is unclear.
The key trade-off is between the reduced number of players admitted into the contest and players' updated beliefs about other players' abilities. For example, if the designer only admits $2$ players among all $100$ players, the two players admitted know their opponents have very high abilities, thus to win the contest, both of them might want to exert more effort. On the other hand, considering there are only $2$ players in the contest might also discourage them from exerting efforts. Without detailed analysis, it is not clear a priori which force plays a more important role. 

\smallskip
 
\noindent\textbf{Model and Results.} 
To address the above question, we model the contest as an all-pay auction and analyze a setting in which players---with heterogeneous private abilities drawn independently and identically from a common distribution---compete for prizes by exerting effort.

Ability represents a player’s expertise in the contest domain and determines the cost of exerting effort.  Intuitively, for the same level of effort, a player with higher ability incurs a lower cost. 
The designer then selects a number of top players
among all registered participants.%

\smallskip

The key challenge and novelty of our analysis stem from characterizing and understanding the impact of players' posterior beliefs. Admitted players update their beliefs about others' abilities based on the signal that their own abilities are among the top ones out of all registered players, but they do not observe their exact rankings. Even with i.i.d.\ priors, these posterior beliefs are non-i.i.d.\ and depend on each player’s private ability. Thus, \textit{posterior beliefs are also private}, which distinguishes our analysis from standard all-pay auctions (contests). Despite the non-i.i.d.\ nature of these beliefs, they share a common structure parameterized by each player's private ability, which allows for the existence of a symmetric and strictly increasing equilibrium strategy.

\smallskip

We find that all effects of the elimination mechanism---including the reduced number of admitted players and the resulting posterior beliefs---are captured by an \emph{inflated ability}. This inflated ability is equal to the true ability divided by a term that depends on the numbers of admitted and registered players, as well as the prior distribution of ability. We show that the inflated ability always exceeds the true ability, which justifies the terminology. We further show that the inflated ability decreases as the number of admitted players increases. In other words, when fewer players are admitted, each admitted player perceives her ability to be relatively higher. However, the inflated ability need not be increasing in the true ability, and a symmetric, strictly increasing equilibrium exists if and only if the inflated ability is increasing in the true ability. Under this condition, we explicitly characterize the unique equilibrium strategy.

For ranking-based prize structures with at most two positive prizes, we show that the equilibrium efforts of \textit{all} admitted players decrease as the number of admitted players increases. Consequently, admitting only two players maximizes the highest equilibrium effort. Finally, we extend our model to a two-stage setting in which all players participate in the first stage, and only those with the top first-stage efforts advance to the second stage. In this environment, we show that no symmetric and strictly increasing perfect Bayesian equilibrium (PBE) exists, even though players are ex-ante identical. This arises because players’ second-stage beliefs depend on their first-stage efforts.

\medskip

\noindent\textbf{Related Literature.} Our paper contributes to the literature on optimal contest design under incomplete information. Much of the existing work focuses on how designers should optimally structure prizes. \citet{Moldovanu_2001_AER_WTA_Optimial} pioneer the analysis of optimal prize structures and show that, when maximizing expected total effort, a winner-take-all format is optimal among all ranking-based prize structures with nonnegative prizes. \citet{liu2018_jet_negative_prizes} and \citet{Liu_2022_optimal_reward_negativeprizes} extend this analysis to settings with negative prizes. \citet{liu2014effort_multiple_prizes} adopt a mechanism-design perspective, showing that under certain conditions a grand all-pay auction is optimal. \citet{Zhang_2024_WTA_convex_fun} study a similar environment under a fixed budget and convex disutility of effort, again identifying an all-pay mechanism as optimal. \citet{Jason_Optimal_Crowdsourcing_Contest} analyze mechanisms that maximize the expected highest effort, and \citet{goel_2023_ec_prizes} examine how prize structures and ability distributions jointly influence equilibrium efforts. In contrast to this prize-design perspective, our paper investigates how the designer can instead leverage the number of admitted players. For related work on all-pay contests with complete information, see \citet{baye1996_all_pay_complete}, \citet{barut1998_all_pay_complete}, and \citet{Baye_2012_CompleteInformation_Rank-based}; for empirical studies, see \citet{sheremeta_geb_2010_experimental_contests}, \citet{tracy_ms_2014_crowdsourcing}, and \citet{tracy_ijgt_2018_experiment_all_pay_auctions}.

\smallskip

Our paper is also related to the literature on multi-stage contests. {\color{black}\citet{2014_Segev_ejor_Equilibrium_Sequential_Incomplete} study multi-stage all-pay contests in which players arrive sequentially and later-arriving players can observe the actions of earlier ones.} \citet{kubitz_2023_aerinsights_two_stage} examine a setting related to our two-stage model, but there are several important differences: (1) their model features two players with binary private types; (2) there is no elimination after the first stage; and (3) first-stage outcomes (such as efforts) are publicly observed, whereas in our setting admitted players only observe whether they are among the top performers.

\smallskip

Furthermore, our paper contributes to the literature on signaling in contests. In our model, the entry restriction serves as a signal about players’ private abilities. The literature on signaling and information disclosure in all-pay contests is rapidly expanding. {\color{black}For instance, \citet{lu_2018_ranking_disclosure} examine settings in which the contest designer has access to players’ private abilities and compare several ranking-disclosure mechanisms.} \citet{ely_ec_2021_optimal_feedback} characterize the optimal feedback policy in contests where effort is monitored through coarse, binary signals. \citet{ersin_2023_or_dynamic_contest} study dynamic development contests that use rank-based incentives and strategic information design to intensify competition among suppliers, with the goal of reducing project lead times. Works such as \citet{chen_2017_persuasion_contests}, \citet{zhao_2023_ridge_infodesign_contests}, and \citet{chen_2024_geb_optimal_disclosure_contests} apply the Bayesian persuasion framework to analyze optimal information disclosure strategies in all-pay contests.

\smallskip

Finally, our paper is closely related to the literature on elimination contests. Similar to our two-stage model, \citet{Mendel_2021_sec_complete_info} study all-pay sequential elimination contests (SEC) under complete information, where players’ abilities are common knowledge. \citet{Fu_2012_Optimal_Elimination_Lottery} and \citet{fu_aermicro_2022_disclosure_sec} examine SECs in the form of Tullock contests, where players are homogeneous in their abilities and the probability of winning or advancing increases with effort. To the best of our knowledge, ours is the first paper to analyze two-stage SECs under incomplete information.\footnote{\citet{2021_Two-Stage} analyze a similar two-stage SEC to ours, but their analysis contains several errors. In particular, they use an incorrect posterior-belief structure and mistakenly assume that players’ posterior beliefs are i.i.d.} The seminal work of \citet{Moldovanu_2006_Contest_Architecture_Sub-elimination} studies a different but closely related two-stage \emph{sub-elimination} contest. In their setting, players are randomly assigned to subgroups in the first stage, and the winners of each subgroup advance to compete in the second stage. The belief structure in sub-elimination contests is relatively simple: a second-stage player perceives her opponents’ abilities as the largest order statistics from each subgroup, and these beliefs are i.i.d. In contrast, in our two-stage model, players’ posterior beliefs are non-i.i.d.\ and depend on their private abilities. Moreover, in two-stage SECs, posterior beliefs are intertwined with players’ first-stage efforts, making the analysis substantially more challenging and novel.

\medskip

\noindent\textbf{Organization.} Section~\ref{sec_model setup} introduces the model setup and defines the equilibrium notion. Section~\ref{sec:posterior_beliefs} characterizes the posterior beliefs. \Cref{sec:warm_up_wta} provides a warm-up analysis focusing on the winner-take-all prize structure with two admitted players. Section~\ref{sec:equilibriumeffort_optimalnumber} derives the equilibrium effort and characterizes the optimal number of admitted players. Section~\ref{sec_two-stage sec} extends the model to a two-stage setting and shows that no symmetric, strictly increasing PBE exists. We conclude in Section~\ref{sec_closing}. All proofs and auxiliary results appear in the Appendix.

\section{Model and Preliminaries}
\label{sec_model setup}

Consider an all-pay contest, where $n_1$ players register. The value of $n_1$ is assumed to be exogenous and commonly known. 
Each player has a private \textit{ability} $a_i$. At the registration stage, each player perceives other players' abilities as i.i.d. random variables drawn from a commonly known distribution function $F(\cdot)$ with continuous and strictly positive
density $f(\cdot)$, noted as the prior beliefs. Motivated by the examples mentioned in \Cref{sec:intro}, we assume that the designer knows the ranking of all players' abilities, potentially inferred from the information provided in their registration materials. By strategically selecting those with the top $n_2\leq n_1$ abilities to admit, the designer aims to 
maximize players' equilibrium efforts. %
The admitted players \emph{only} observe who are admitted and who are not. We assume $n_2\geq 2$ to avoid trivial results.

We assume that the prize structure is exogenously given.
One of the commonly used prize structures is ranking-based---the allocation of prizes only depends on the ranking of players' efforts (see, e.g., \citealt{Moldovanu_2001_AER_WTA_Optimial,Moldovanu_2006_Contest_Architecture_Sub-elimination,Liu_2022_optimal_reward_negativeprizes}). 
Specifically, let $L\in \mathbb{Z}_{>0}$ ($L\le n_2$) be the number of total prizes, and the value of the $\ell^{\textrm{th}}$ prize is denoted by $V_\ell$ with $V_1 \geq V_2 \geq \dots \geq V_L \geq 0$. For a given $n_2$, without loss of generality, we can let $L=n_2$ because some prizes can be $0$. 
Let $V=[V_{\ell}]_{\ell=1}^{n_2}$.
We will briefly discuss how to set up the optimal prize structure in \Cref{sec:equilibriumeffort_optimalnumber}, although this is not the focus of our paper.

\smallskip

Denote by $\mathcal{I}$ the set of $n_2$ admitted players. 
Let $A_i$ be the random variable (perceived by other players) of player $i$'s ability and $a_i$ be its realization. Without loss of generality, we assume that $a_i\in[0,1]$. 
A player $i \in \mathcal{I}$ exerts an effort $e_i$ which incurs a \textit{cost} $g(e_i)/a_i$ where $g: \mathbb R_{\ge0} \mapsto \mathbb R_{\ge0}$ is a strictly increasing, continuous and differentiable function with $g(0)=0$. Define $g(0)/0:=0$. Given the same effort, the player with higher ability incurs lower costs. Informally, the ability parameter $a_i$ can be explained as the rate at which player $i$ works.  
The player with the highest effort wins the first prize, $V_1$; similarly, the player with the second-highest effort wins the second prize, $V_2$, and so on until all the prizes are allocated.

Put $e=[e_i: i\in \mathcal{I}]\in \mathbb{R}_{\ge0}^{n_2}$ as all admitted players' efforts. Then, player $i$'s ex-post utility under effort profile $e$ is
\begin{align*}
    u_i(e_i,e_{-i})= \sum_{\ell =1}^L V_\ell\cdot \mathbbm{1}\{\text{$e_i$ is the $\ell^\textrm{th}$ highest among $e$}\} - \frac{g(e_i)}{a_i},
\end{align*}
where $\mathbbm{1}\{\cdot\}$ is the indicator function and ties are broken arbitrarily. %

\medskip

\noindent\textbf{Strategies.} 
Following the standard assumption in all-pay contest literature \citep{Moldovanu_2001_AER_WTA_Optimial,Moldovanu_2006_Contest_Architecture_Sub-elimination,Liu_2022_optimal_reward_negativeprizes}, we focus on symmetric and strictly increasing
strategies, where all players exert efforts according to their abilities, and the higher the ability, the more the effort.
Specifically, for any player $i\in \mathcal{I}$, we assume that her effort $e_i$ is a function of her ability $e_i = b(a_i)$ where $b(\cdot)$ is differentiable and strictly increasing in $a_i$.

\smallskip

\noindent\textbf{Belief System.} After being admitted into the contest, the player knows that her opponents are among the top $n_2$ abilities in the $n_1$ players. We assume that all players are Bayesian and will update their beliefs from the \textit{admission signal}. Let $s_i\in\{0,1\}$ be player $i$'s admission signal. Define $s_i=1$ if player $i$ is admitted into the contest and $0$ otherwise. Let $s = (s_1,s_2,\cdots, s_{n_1})$ be all $n_1$ players' signals, and $s$ is common knowledge for all players. 
Denote by $s_{-i}$ all other $n_1-1$ players' signals except that of player $i$ and $\mathcal{I}_{-i}$ as all admitted players except player $i$. Let $a_{-i} = \left[a_j : j\in\mathcal{I}_{-i} \right]$ be all other admitted players' abilities except that of player $i$, and let $A_{-i}$ be the corresponding random variables. Define $\beta_i(a_{-i}\mid s, a_i)$ as admitted player $i$'s posterior belief (joint probability density) about all other admitted players' abilities conditional on player $i$'s ability and the admission signal $s$. 
Note that, $\beta_i(a_{-i}\mid s,a_i)=\beta_i(a_{-i}\mid s,a_i,n_2)$ since the signal profile $s$ contains the information about the admitted number $n_2$. Thus, we drop the condition $n_2$ in the belief, and whether the player knows the admitted number in the posterior belief depends on whether it is conditioned on $s$.

\smallskip

\noindent\textbf{Equilibrium.} We adopt the notion of Bayesian Nash equilibrium (BNE) to define our solution concept. A BNE is a tuple of strategy $b(\cdot)$ and posterior beliefs $\left[\beta_i(a_{-i}\mid s,a_i): i \in \mathcal{I}\right]$ that satisfies the following conditions:

\begin{enumerate}
\item[(i)] For every player $i \in \mathcal{I}$, Bayes' rule is used to update her posterior belief $\beta_i(a_{-i}\mid s,a_i)$,
    \begin{align}
    \beta_i(a_{-i} \mid s ,a_i) = \beta_i(a_{-i} \mid s_{-i},s_{i},a_i) & =  \frac{\Pr\left(s_{-i}\mid a_{-i},s_{i}, a_i\right)\prod_{j\in\mathcal{I}_{-i}}f(a_j)}{\int_{a_{-i}}\Pr\left(s_{-i} \mid a_{-i},s_{i}, a_i\right)\prod_{j\in\mathcal{I}_{-i}}f(a_j)da_{-i}} \label{eq_posterior_bayes_rule}. 
\end{align}
The term $\Pr\left(s_{-i}\mid a_{-i},s_{i}, a_i\right)$ is the probability that the other $n_2-1$ players get admitted into the contest conditional on player $i$ being admitted, her ability $a_i$, and the other $n_2-1$ players' abilities $a_{-i}$. The formal proof for equation (\ref{eq_posterior_bayes_rule}) is provided in the proof of Proposition \ref{prop_post_belief} in the appendix.

\smallskip

\item[(ii)] For every player $i \in \mathcal{I}$,
\begin{align*}
    b(a_i) \in  {\arg\max}_{e_i}~ \sum_{\ell=1}^L V_\ell P_{i,\ell}(e_i\mid b) - \frac{g\left(e_i\right)}{a_i},
\end{align*}
where $P_{i,\ell}(e_i\mid b)$ is the probability of a player $i \in \mathcal{I}$ winning the prize by exerting effort $e_i$ given all other admitted players follow the strategy $b(\cdot)$ based on her posterior belief. Formally,
\begin{align*}
    P_{i,\ell}(e_i\mid b):= \Pr \left( \text{$e_i$ ranks the highest in $\left\{b(A_j): j\in\mathcal{I}_{-i}\right\}\cup\left\{e_i\right\}$}\right), \quad A_{-i} \sim \beta_i(a_{-i}\mid s,a_i).
\end{align*}
\end{enumerate}

\medskip

\noindent\textbf{Timeline.} 
The timeline of events is as follows: an exogenous number of \(n_1\) players register for the contest. 
The designer then sends a signal \(s\) to all players specifying which $n_2$ players are admitted. 
After receiving this signal, admitted players update their beliefs and exert effort to compete for prizes.

\medskip
\noindent\textbf{Designer's problem.} The designer aims to maximize the expected highest effort by choosing the admitted number of players $n_2$. The highest effort is a commonly used performance metric in the contest literature \citep{Moldovanu_2001_AER_WTA_Optimial, Moldovanu_2006_Contest_Architecture_Sub-elimination, Jason_Optimal_Crowdsourcing_Contest}.

\medskip

\noindent \textbf{Notation.} For any $x,y\in \mathbb{R}$, we use $x\wedge y:=\min\{x,y\}$ to denote the minimum of two values.
When we say a function is ``increasing/decreasing'', it means ``weakly increasing/decreasing''.

\section{Posterior Beliefs}
\label{sec:posterior_beliefs}
In this section, we characterize the admitted players' posterior beliefs $\beta_i(\cdot \mid s, a_i)$.

Define the incomplete beta function $B(x,p,q)$ with parameters $x\in[0,1], p\in\mathbb Z_{\ge0}, q\in\mathbb Z_{\ge0}$ as
\[
B(x,p,q) := \int_0^x t^{p-1} (1-t)^{q-1}dt,
\]
and define $B(x,p,q):=0, \forall x\in[0,1]$ when $pq=0$. Furthermore, for $x\in[0,1]$, define
\[
I(x,n_2) := x^{n_1-n_2}\bigl(1 - x \bigl)^{n_2-1}  + (n_2-1) B\bigl(x,n_1-n_2+1,n_2-1\bigl).
\]

In the following, whenever posterior beliefs are mentioned, these are the beliefs of admitted players as eliminated players do not affect the contest outcome.

\begin{proposition}[Posterior Beliefs]
\label{prop_post_belief}
If player $i$ with ability $a_i$ is admitted into the contest, her belief (joint density) about the other $n_2-1$ admitted players' abilities $A_{-i}$ is
\begin{equation}
\label{Eq_Posterior beliefs}
\beta_i(a_{-i} \mid s ,a_i) = \frac{F^{n_1-n_2}\left(a_i\wedge\normalfont{\min}_{j\in \mathcal{I}_{-i}} a_j\right)}{I(F(a_i),n_2)}\prod\limits_{j \in \mathcal{I}_{-i}}f(a_j),~ a_{-i}\in[0,1]^{n_2-1}.
\end{equation}
\end{proposition}
\begin{proof}[Proof Sketch]
We provide a sketch of the proof of \Cref{prop_post_belief} and highlight a common mistake.

To derive the posterior beliefs, it remains to determine $\Pr(s_{-i} \mid a_{-i}, s_i, a_i)$ in \eqref{eq_posterior_bayes_rule}. This term represents the probability that, given player~$i$'s ability $a_i$, the event that player~$i$ is admitted ($s_i = 1$), and the abilities $a_{-i}$ of the other $n_2 - 1$ players, those players are indeed admitted. Let $a_{(1)}$ be the minimum ability among these $n_2-1$ players, and let $\Bar{\mathcal{I}}$ be the set of eliminated players (all players except for player $i$ and these $n_2-1$ players). Then, we have
\begin{align}
    \Pr(s_{-i} \mid a_{-i}, s_i, a_i) = 
    \begin{cases}
    1, & \text{if } a_{(1)} > a_i, \\
    \Pr\left( \max_{k \in\Bar{\mathcal{I}}} A_k < a_{(1)} \mid a_{-i}, s_i, a_i \right), & \text{if } a_{(1)} < a_i.
    \end{cases}
    \label{eq:prob_admission_signal}
\end{align}
If $a_{(1)} > a_i$, then since player $i$ is admitted, these $n_2-1$ players must also be admitted, so the probability is 1. If $a_{(1)} < a_i$, to ensure these $n_2-1$ players are admitted, we must have all eliminated players' abilities smaller than $a_{(1)}$. This leads to the second part of \eqref{eq:prob_admission_signal}.

A \emph{naive} line of reasoning might suggest that $A_k$ for eliminated players still follows the prior distribution. This is incorrect because player $i$ knows that all eliminated players have abilities smaller than $a_i$. Thus, from player $i$'s perspective, $A_k$ is drawn from the distribution with CDF:
\[
    H(a_k) := \Pr\{ A_k < a_k \mid A_k < a_i \} = \frac{\Pr\{ A_k < a_k, A_k < a_i \}}{\Pr\{ A_k < a_i \}} = \frac{F(a_k)}{F(a_i)}, \quad a_k \in [0, a_i].
\]
The corresponding PDF is $h(a_k) = f(a_k)/F(a_i)$. Thus,
\begin{align}
    &\Pr\left( \max_{k \in \Bar{\mathcal{I}}} A_k < a_{(1)} \mid a_{-i}, s_i, a_i \right) = \prod_{k \in \Bar{\mathcal{I}} } \int_0^{a_{(1)}} \frac{f(a_k)}{F(a_i)} \, da_k  = \frac{F^{n_1 - n_2}(a_{(1)})}{F^{n_1-n_2}(a_i)}. \nonumber
\end{align}

In the formal proof, we first derive player $i$'s belief about all other players' abilities, including eliminated players. By integrating over the eliminated players' abilities, we then obtain player $i$'s belief about the abilities of the other admitted players.
\end{proof}

As a special case of the posterior belief with $n_1=n_2$, for every $a_i\in [0,1]$, we have
\begin{align*}
    I(F(a_i),n_1) = F^{n_1-n_1}(a_i)\cdot (1-F(a_i))^{n_1-1} + (n_1-1)\cdot \int_{0}^{F(a_i)} t^{n_1-n_1+1-1}(1-t)^{n_1-1-1}dt = 1,
\end{align*}
and thus $\beta_i(a_{-i}\mid s,a_i) = \prod_{j\in\mathcal{I}_{-i}}f(a_j)$. In other words, when all $n_1$ players are admitted, every player's posterior belief is the same as her prior belief.

\smallskip

Proposition~\ref{prop_post_belief} shows that the posterior beliefs about a given player's ability are \emph{not} identical across players; rather, they depend on each player's own private ability. In other words, \emph{players' posterior beliefs are themselves private}. Moreover, $\beta_i(a_{-i}\mid s, a_i)$ can be \emph{non-differentiable} (though still continuous) at the set $\{a_{-i} : \min_{j\in\mathcal{I}_{-i}} a_j = a_i\}$. Both the dependence on $a_i$ and the possible non-differentiability arise from the fact that whenever a player $j$'s ability is strictly higher than $a_i$, player~$i$ knows that, \emph{with probability one}, player~$j$ is admitted, conditional on player~$i$ herself being admitted.

\smallskip
A direct corollary from Proposition \ref{prop_post_belief} is the following result about players' marginal posterior beliefs, i.e., the posterior belief about one single admitted opponent's ability. 
\begin{corollary}[Marginal Posterior Beliefs]\label{Marginal Posterior Belief}
Player $i$'s ($i \in \mathcal{I}$) belief about player $j$'s ($j \in \mathcal{I}_{-i}$) ability is that, 
\begin{align*}
\beta_i\left(a_j\mid s,a_i \right) = \frac{f(a_j)}{I(F(a_i),n_2)}\Big(&(n_2-2)B\left(F(a_i\wedge a_j),n_1-n_2+1,n_2-2\right)\\
&+F^{n_1-n_2}(a_i\wedge a_j)(1-F(a_i\wedge a_j))^{n_2-2}\Big),~~\forall a_j\in [0,1].
\end{align*}
\end{corollary}
An admitted player's posterior belief about each opponent's ability is identical across opponents. What differs, however, are the posterior beliefs held by two different players about a third player's ability, due to their distinct private abilities. Although the posterior beliefs about different opponents' abilities are identical, they are \emph{not} independent. {\color{black}Consequently, the marginal posterior beliefs are not sufficient; the joint density in \Cref{prop_post_belief} is required to fully characterize the posterior beliefs.}

\begin{remark}
One player's marginal posterior beliefs about different players' abilities can be correlated. Formally, $\beta_i(a_{-i}\mid s,a_i)\neq \prod_{j\in \mathcal{I}_{-i}} \beta_i(a_j\mid s,a_i)$ in general.
\end{remark}

In words, suppose players $i,j,k$ are all admitted into the contest, player $i$'s belief about player $j$'s ability can depend on player $i$'s belief about player $k$'s. This is because when player $i$ gets admitted and knows that player $k$ also gets admitted, if player $j$'s ability is higher than player $k$'s, player $i$ knows that with probability one player $j$ also gets admitted. In other words, for player $i$, knowing player $k$'s ability has an impact on her belief about player $j$'s. It is worth noting that our one-round SEC is not an instance of standard Bayesian games with \emph{common} asymmetric and correlated beliefs as players' posterior beliefs rely on their individual \emph{private} abilities.

Although the same player's posterior beliefs about different opponents' abilities are correlated, different players' posterior beliefs about the same player's ability are independent, i.e., $\beta_i(a_j \mid s, a_i)$ and $\beta_k(a_j \mid s, a_k)$ are independent for any $i\neq j\neq k\in \mathcal{I}$. More generally, $\beta_i(a_{-i}\mid s,a_i)$ is independent of $\beta_j(a_{-j}\mid s,a_j)$, for any admitted players $i\neq j$.  This independence helps to simplify the analysis. 

\begin{example}
We now give an example with $n_2=2$ to illustrate the posterior beliefs derived in Proposition \ref{prop_post_belief}. When $n_2=2$, for any player $i\neq j \in \mathcal{I}$,
\begin{align*}
    \label{eq_poster_n_2=2_example}
    \beta_i\left(a_j \mid s,a_i\right) = 
\frac{f(a_j)}{\frac{F^{n_1-1}\left(a_i\right)}{n_1-1} + F^{n_1-2}(a_i)(1 - F(a_i))}F^{n_1-2}\left(a_i \wedge a_j\right),~~\forall a_j\in[0,1].
\end{align*}

Figure \ref{Fig_Comparison Beliefs} shows the comparison between a uniform prior belief and its corresponding posterior beliefs based on different values of $a_i$ with $n_1=5, n_2=2$. Since we assume a uniform prior, when $a_j>a_i$, the posterior density becomes a constant. Observe that the posterior density is not differentiable at $a_j=a_i$, i.e., there is a jump of the derivative of posterior density at $a_j=a_i$.
The amount of this jump can measure the impact of the admission signal on players' posterior beliefs. It
is not monotone in $a_i$ in general. One way to think about this is that the admission signal has little impact on the two extreme players with $a_i=0$ or $a_i=1$ since they know their abilities are lower or higher than any other players. 
On the other hand, our results (see Appendix \ref{sec_appendix:numerical}) suggest that the amount of this jump is non-increasing in $n_2$. This is also intuitive because  
the closer $n_2$ and $n_1$ are, the less the impact on players' posterior beliefs. One extreme case is $n_2=n_1$, where the posterior beliefs are the same as prior beliefs and there is no jump in this situation.

\smallskip

Another observation from the posterior CDF is that no matter what value $a_i$ takes,  once advancing to the contest, player $i$ always perceives her opponent stronger compared to her prior belief, i.e., $\: \Pr_{A_j\sim \beta_i(\cdot \mid s, a_i)}\left(A_j > a_j\right) > \Pr_{A_j \sim f(\cdot)}\left(A_j > a_j\right), \forall a_j \in [0,1]$. In fact, this observation is true in general as stated in the following proposition.

\begin{figure*}[htbp]
  \centering
  \subfloat[PDF: $a_i=0.3$]{  \resizebox{.32\textwidth}{!}{\begin{tikzpicture}
\begin{axis}[
    xlabel={$a_j$},
    ylabel={Density},
    legend style={at={(0.05,0.8)}, anchor=west}, %
    thick,
    ymin=0, %
    ymax=2.7,
    xmin=0,
    xmax=1,
    no markers,
    every axis plot post/.append style={line width=1.8pt}, %
    tick style={line width=1pt, black}, %
    tick label style={font=\large, scale=1.2}, %
    label style={font=\large, scale=1.2}
]
\addplot[red, dashed] table[x=a_j_values, y=Prior_PDF] {fig/belief/data/pdf_data_a_i=0.30_n_1=5.txt};
\addlegendentry{Prior PDF}

\addplot[blue] table[x=a_j_values, y=Posterior_PDF] {fig/belief/data/pdf_data_a_i=0.30_n_1=5.txt};
\addlegendentry{Posterior PDF}

\draw[dotted, line width=1.5pt] (axis cs:0.3,0) -- (axis cs:0.3,1.3); %
\end{axis}
\end{tikzpicture}}}
 \hfill
 \subfloat[PDF: $a_i=0.5$]{  \resizebox{.32\textwidth}{!}{\begin{tikzpicture}
\begin{axis}[
    xlabel={$a_j$},
    ylabel={Density},
    legend style={at={(0.05,0.8)}, anchor=west}, %
    thick,
    ymin=0, %
    ymax=2.7,
    xmin=0,
    xmax=1,
    no markers,
    every axis plot post/.append style={line width=1.8pt}, %
    tick style={line width=1pt, black}, %
    tick label style={font=\large, scale=1.2}, %
    label style={font=\large, scale=1.2}
]
\addplot[red, dashed, thick] table[x=a_j_values, y=Prior_PDF] {fig/belief/data/pdf_data_a_i=0.50_n_1=5.txt};
\addlegendentry{Prior PDF}

\addplot[blue, thick] table[x=a_j_values, y=Posterior_PDF] {fig/belief/data/pdf_data_a_i=0.50_n_1=5.txt};
\addlegendentry{Posterior PDF}

\draw[dotted, line width=1.5pt] (axis cs:0.5,0) -- (axis cs:0.5,1.6);
\end{axis}
\end{tikzpicture}}}
  \hfill
 \subfloat[PDF: $a_i=0.8$]{  \resizebox{.32\textwidth}{!}{\begin{tikzpicture}
\begin{axis}[
    xlabel={$a_j$},
    ylabel={Density},
    legend style={at={(0.05,0.8)}, anchor=west}, %
    thick,
    ymin=0, %
    ymax=2.7,
    xmin=0,
    xmax=1,
    no markers,
    every axis plot post/.append style={line width=1.8pt}, %
    tick style={line width=1pt, black}, %
    tick label style={font=\large, scale=1.2}, %
    label style={font=\large, scale=1.2}
]
\addplot[red, dashed, thick] table[x=a_j_values, y=Prior_PDF] {fig/belief/data/pdf_data_a_i=0.80_n_1=5.txt};
\addlegendentry{Prior PDF}

\addplot[blue, thick] table[x=a_j_values, y=Posterior_PDF] {fig/belief/data/pdf_data_a_i=0.80_n_1=5.txt};
\addlegendentry{Posterior PDF}

\draw[dotted, line width=1.5pt] (axis cs:0.8,0) -- (axis cs:0.8,2.6);
\end{axis}
\end{tikzpicture}}}
\vspace{0.9em}
 \subfloat[CDF: $a_i=0.3$]{  \resizebox{.32\textwidth}{!}{\begin{tikzpicture}
\begin{axis}[
    xlabel={$a_j$},
    ylabel={Probability},
    legend style={at={(0.05,0.8)}, anchor=west}, %
    thick,
    ymin=0, %
    ymax=1.1,
    xmin=0,
    xmax=1,
    no markers,
    every axis plot post/.append style={line width=1.8pt}, %
    tick style={line width=1pt, black}, %
    tick label style={font=\large, scale=1.2}, %
    label style={font=\large, scale=1.2}
]
\addplot[red, dashed] table[x=a_j_values, y=Prior_CDF] {fig/belief/data/cdf_data_a_i=0.30_n_1=5.txt};
\addlegendentry{Prior CDF}

\addplot[blue] table[x=a_j_values, y=Posterior_CDF] {fig/belief/data/cdf_data_a_i=0.30_n_1=5.txt};
\addlegendentry{Posterior CDF}

\draw[dotted, line width=1.5pt] (axis cs:0.3,0) -- (axis cs:0.3,1.3); %
\end{axis}
\end{tikzpicture}}}
  \hfill
   \subfloat[CDF: $a_i=0.5$]{  \resizebox{.32\textwidth}{!}{\begin{tikzpicture}
\begin{axis}[
    xlabel={$a_j$},
    ylabel={Probability},
    legend style={at={(0.05,0.8)}, anchor=west}, %
    thick,
    ymin=0, %
    ymax=1.1,
    xmin=0,
    xmax=1,
    no markers,
    every axis plot post/.append style={line width=1.8pt}, %
    tick style={line width=1pt, black}, %
    tick label style={font=\large, scale=1.2}, %
    label style={font=\large, scale=1.2}
]
\addplot[red, dashed] table[x=a_j_values, y=Prior_CDF] {fig/belief/data/cdf_data_a_i=0.50_n_1=5.txt};
\addlegendentry{Prior CDF}

\addplot[blue] table[x=a_j_values, y=Posterior_CDF] {fig/belief/data/cdf_data_a_i=0.50_n_1=5.txt};
\addlegendentry{Posterior CDF}

\draw[dotted, line width=1.5pt] (axis cs:0.5,0) -- (axis cs:0.5,1.3); %
\end{axis}
\end{tikzpicture}}}
  \hfill
   \subfloat[CDF: $a_i=0.8$]{  \resizebox{.32\textwidth}{!}{\begin{tikzpicture}
\begin{axis}[
    xlabel={$a_j$},
    ylabel={Probability},
    legend style={at={(0.05,0.8)}, anchor=west}, %
    thick,
    ymin=0, %
    ymax=1.1,
    xmin=0,
    xmax=1,
    no markers,
    every axis plot post/.append style={line width=1.8pt}, %
    tick style={line width=1pt, black}, %
    tick label style={font=\large, scale=1.2}, %
    label style={font=\large, scale=1.2}
]
\addplot[red, dashed] table[x=a_j_values, y=Prior_CDF] {fig/belief/data/cdf_data_a_i=0.80_n_1=5.txt};
\addlegendentry{Prior CDF}

\addplot[blue] table[x=a_j_values, y=Posterior_CDF] {fig/belief/data/cdf_data_a_i=0.80_n_1=5.txt};
\addlegendentry{Posterior CDF}

\draw[dotted, line width=1.5pt] (axis cs:0.8,0) -- (axis cs:0.8,1.3); %
\end{axis}
\end{tikzpicture}}}
  \caption{Comparison between prior and posterior beliefs ($n_1=5$, $n_2=2$ with uniform prior distribution)
  }
\label{Fig_Comparison Beliefs}
\end{figure*}

\end{example}

\begin{proposition}[Stochastic Dominance of Posterior Belief over Prior Belief]
\label{Proposition_margin_distribution_Stochastic Dominance}
    Any admitted player's posterior belief about the ability of an admitted opponent first-order stochastically dominates the prior belief. Formally, for any $n_2 < n_1$ and for all $a_i\in[0,1]$ and $i\neq j \in \mathcal{I}$,

\begin{equation*}
    \Pr_{A_{j}\sim \beta_i(\cdot\mid s,a_i)}\Big(A_{j} > a_j \Big) \geq  \Pr_{A_{j} \sim f(\cdot)}\Big(A_{j} > a_j\Big),~\forall a_j \in [0,1].
\end{equation*}
\end{proposition}

Stochastic dominance implies that if a player is admitted into the contest, she will perceive the average ability of her opponents to be higher. Formally, $\mathbb{E}_{A_j\sim \beta_i(\cdot\mid s,a_i)}\left[A_j\right] \geq \mathbb{E}_{A_j\sim f(\cdot)}\left[A_j\right], \forall a_i \in [0,1]$. %

\section{Warm Up: Winner-Take-All Prize and $n_2=2$}
\label{sec:warm_up_wta}

In this section, for illustration, we consider a simple case with a winner-take-all prize structure, i.e., $V_1=1$, only two players are admitted, and the cost function is linear: $g(x)=x$. We index the two admitted players as $i$ and $j$.

When player $j$ follows a strictly increasing strategy $b(\cdot)$ and player $i$ exerts effort $b(\tilde{a}_i)$ for some $\tilde{a}_i\in[0,1]$, the utility of player $i$ (excluding the factor $1/a_i$) is
\begin{align*}
    U_i(\tilde{a}_i,a_i\mid n_2, b) 
    := ~& a_i \cdot \Pr_{A_{j}\sim \beta_i(\cdot \mid s,a_i)}\{b(A_{j})<b(\tilde{a}_i)\} - b(\tilde{a}_i) \\
     = ~&a_i \cdot \Pr_{A_{j}\sim \beta_i(\cdot \mid s,a_i)}\{A_{j}<\tilde{a}_i\} - b(\tilde{a}_i).
\end{align*}
By \Cref{prop_post_belief} regarding the posterior belief $\beta_i(\cdot \mid s,a_i)$, we have
\begin{align}
\label{eq:utility_wta}
    U_i(\tilde{a}_i,a_i\mid n_2, b) =
    \begin{cases}
      \frac{a_i}{J(F(a_i),n_2)} \cdot F^{n_1-1}(\tilde{a}_i) - b(\tilde{a}_i), & \quad \text{if } \tilde{a}_i\in[0,a_i],\\
      \frac{a_i}{J(F(a_i),n_2)}\cdot F^{n_1-2}(a_i) \bigl[
      (n_1-1) F(\tilde{a}_i) - (n_1-2) F(a_i)
      \bigr] - b(\tilde{a}_i),    & \quad \text{if } \tilde{a}_i\in[a_i,1],
    \end{cases}
\end{align}
where $J(x,n_2):= \binom{n_1-1}{n_2-1} \cdot I(F(a_i),n_2)$ for $x\in[0,1]$. We have $J(x,n_1)=1$ for all $x\in[0,1]$.
Notice that the only term depending on $n_2$ is the denominator $J(F(a_i),n_2)$.

To make $b(\cdot)$ an equilibrium, $U_i(\tilde{a}_i,a_i\mid n_2, b)$ must be maximized at $\tilde{a}_i=a_i$. In other words, when player $j$ follows $b(\cdot)$, it should be optimal for player $i$ to do so as well. The equilibrium utility $U_i(a_i,a_i\mid n_2, b)$ is
\begin{align}
\label{eq:equilibrium_utility_n2_b}
 U_i(a_i,a_i\mid n_2, b)   
 = 
\underbrace{\frac{a_i}{J(F(a_i),n_2)}}_{:=v(a_i\mid n_2)} \cdot F^{n_1-1}(a_i) - b(a_i).
\end{align}
For a regular all-pay contest without elimination (i.e., $n_2=n_1$), we have
\begin{align}
\label{eq:equilibrium_utility_n1_b}
    U_i(a_i,a_i\mid n_1,a_i) = a_i \cdot F^{n_1-1}(a_i) - b(a_i).
\end{align}

The only difference between \eqref{eq:equilibrium_utility_n2_b} and \eqref{eq:equilibrium_utility_n1_b} is that under elimination, player $i$ pretends to have an \textit{inflated ability} $v(a_i\mid n_2)$ at equilibrium.
We call it ``inflated'' ability since $v(a_i\mid n_2)\geq a_i$ as shown later in \Cref{sec:equilibriumeffort_optimalnumber}.
Notice that, $v(a_i\mid n_2)$ may \textit{not} be increasing in $a_i\in[0,1]$.

To make $b(\cdot)$ an equilibrium, it must satisfy the first-order condition:
\begin{align*}
  \frac{\partial U_i(\tilde{a}_i,a_i\mid n_2, b)}{\partial \tilde{a}_i}\Big|_{\tilde{a}_i=a_i} = 0 
  \quad \text{for all } a_i\in[0,1].
\end{align*}
This gives a unique solution: for all $a_i\in[0,1]$,
\begin{align}
\label{eq_b(a_i)_wta}
 b(a_i) = \int_{0}^{a_i} v(x\mid n_2)\, dF^{n_1-1}(x).
\end{align}

However, the first-order condition may \emph{not} yield a global maximum.
To make \eqref{eq_b(a_i)_wta} an equilibrium, we need 
\begin{align}
\label{eq:equilibriumcondition_wta}
 a_i \in \argmax_{\tilde{a}_i\in [0,1]}~ U_i(\tilde{a}_i,a_i\mid n_2,b),
\end{align}
where the utility $U_i(\tilde{a}_i,a_i\mid n_2,b)$ is given by \eqref{eq:utility_wta} with strategy $b(\cdot)$ in \eqref{eq_b(a_i)_wta}. 
We claim that \eqref{eq_b(a_i)_wta} is indeed an equilibrium if the inflated ability $v(a_i\mid n_2)$ is increasing in $a_i\in [0,1]$.
To see this, consider
the derivative of the utility with respect to $\tilde{a}_i$:
\begin{align*}
&\textrm{if $\tilde{a}_i\in [0,a_i]$,} \quad  
\frac{\partial U_i(\tilde{a}_i,a_i\mid n_2,b)}{\partial \tilde{a}_i}
=\left(F^{n_1-1}(\tilde{a}_i)\right)^\prime 
   \cdot \left[v(a_i\mid n_2)  - v(\tilde{a}_i\mid n_2)\right]; \\
& \textrm{if $\tilde{a}_i\in [a_i,1]$,} \quad   
   \frac{\partial U_i(\tilde{a}_i,a_i\mid n_2,b)}{\partial \tilde{a}_i}=(n_1-1)\cdot f(\tilde{a}_i) \cdot \left[
  v(a_i\mid n_2)\cdot F^{n_1-2}(a_i) 
   - v(\tilde{a}_i\mid n_2)\cdot F^{n_1-2}(\tilde{a}_i)  
  \right] .
\end{align*}

If the inflated ability $v(a_i\mid n_2)$ is increasing in $a_i\in [0,1]$, it is easy to show that the derivative $\partial U_i(\tilde{a}_i,a_i\mid n_2,b)/\partial \tilde{a}_i$ is non-negative in the range $\tilde{a}_i\in [0,a_i]$ and non-positive in the range $\tilde{a}_i\in [a_i,1]$.
That is, the utility $ U_i(\tilde{a}_i,a_i\mid n_2,b)$ is increasing in $\tilde{a}_i\in [0,a_i]$ and decreasing in $\tilde{a}_i\in [a_i,1]$. 
Thus, condition \eqref{eq:equilibriumcondition_wta} holds.
We will show later that this observation holds generally, and that the increasing property of $v(\cdot \mid n_2)$ is also \textit{necessary} for the existence of a symmetric and strictly increasing equilibrium.

\section{Equilibrium Efforts and Optimal Admitted Number}
\label{sec:equilibriumeffort_optimalnumber}
We now turn to a general setting. 
Denote by $F_{(\ell,n)}(\cdot)$ the distribution function of the $\ell^{\textrm{th}}$ largest order statistic among $n$ i.i.d. random variables with distribution $F(\cdot)$. We have the following identities, for $1\leq \ell \leq n$,
\begin{align*}
F_{(\ell,n)}(x) =& \sum_{j=n-\ell+1}^n \binom{n}{j}F^j(x)\bigl(1 - F(x) \bigl)^{n-j},\\
dF_{(\ell,n)}(x) =& \frac{n!}{(n-\ell)!(\ell-1)!}F^{n-\ell}(x)\bigl(1 - F(x) \bigl)^{\ell-1}dF(x).
\end{align*}
We further define $F_{(0,n)}(x)=0, F_{(n,n-1)}(x)=0,dF_{(0,n)}(x)=0dF(x),\forall x \in[0,1]$.

Put $g^{-1}(\cdot)$ to be the inverse function of the cost function $g(\cdot)$.
For $x\in[0,1]$ and $n_2\in[2,n_1]$, define
\[
J(x,n_2) := \binom{n_1-1}{n_2-1} \cdot I(x,n_2).
\]
Note that $J(x,n_1)=1$ for any $x\in[0,1]$ since $I(x,n_1)=1$ as mentioned before.

Define the \textit{inflated ability} as follows: for any $a_i\in[0,1]$ and for any $n_2\in[2,n_1]$,
\begin{align*}
    v(a_i\mid n_2) = \frac{a_i}{J(F(a_i),n_2)}.
\end{align*}
Our next result characterizes the properties of the function $J(F(a_i),n_2)$ and shows that the inflated ability $v(a_i\mid n_2)\geq a_i$.
\begin{lemma}[Inflation Effect]
\label{lemma_J_increasing_n_2}
For any $n_1\geq 2$ and any $x\in[0,1]$,
\begin{enumerate}
    \item[(i)] $J(x,n_2)\leq 1$ for any $n_2\in[2,n_1]$;
    \item[(ii)] $J(x,n_2)$ is increasing in $n_2 \in [2,n_1]$.
\end{enumerate}
\end{lemma}

Part~(i) establishes that, for any admitted player \(i\), the inflated ability \(v(a_i \mid n_2)\) strictly exceeds her true ability \(a_i\), since knowing she has outperformed \(n_1 - n_2\) players leads her to update her beliefs upward about her relative standing. Part~(ii) further shows that this inflation effect weakens as the number of admitted players increases: admitting more players provides a weaker informational signal that ``my ability is relatively high,'' whereas admitting fewer players makes each admitted player more certain of belonging to a very exclusive top tier, thereby reinforcing a stronger upward revision of her perceived ability.

Although the inflated ability $v(a_i\mid n_2)$ is decreasing in $n_2$ by \Cref{lemma_J_increasing_n_2} (ii), it may not be monotone in $a_i$.
\begin{proposition}\label{prop:monotonicity of v}
For the power-law prior distribution $F(x)=x^\theta$ with $\theta>0$, given $n_2\in[2,n_1]$,
\begin{enumerate}
    \item[(i)] there exists $\hat{a}\in[0,1)$ such that the inflated ability $v(a_i\mid n_2)$ decreases in $a_i\in [0,\hat{a}]$ and increases in $a_i\in[\hat{a},1]$;
    \item[(ii)] the following statements are equivalent,
    \begin{enumerate}
        \item $v(a_i\mid n_2)$ is increasing with $a_i\in[0,1]$, i.e., $\hat{a}=0$,
        \item $v(0\mid n_2)$ is finite,
        \item $(n_1-n_2)\cdot \theta \leq 1$.
    \end{enumerate}
\end{enumerate}
\end{proposition}

Notice that, for general prior distribution $F(\cdot)$, the finiteness of $v(0\mid n_2)$ is only a necessary condition, but it may not be sufficient to guarantee that $v(a_i\mid n_2)$ is increasing in $a_i\in[0,1]$.\footnote{For example, consider $n_1=3$, $n_2=2$, and a prior distribution $F(x)=0.2x + 0.8x^3$ for $x\in[0,1]$. Under these parameters, $v(0\mid n_2)$ is finite; however, $v(a_i\mid n_2)$ first increases, then decreases, and then increases again in $a_i$.}

Let $b(\cdot \mid n_2)$ denote the equilibrium strategy when the admitted number of players is $n_2$. Our next result characterizes the unique symmetric and strictly increasing strategy (if it exists) and shows that the increasing property of the inflated ability plays a critical role.

\begin{theorem}[Equilibrium Strategy]
\label{thm_equilibrium strategy_general_prize}
    For any $n_2\in[2,n_1]$, 
\begin{enumerate}
    \item[(i)] there exist symmetric and strictly increasing equilibrium strategies \emph{if and only if}
$v(a_i\mid n_2)$ is non-decreasing in $a_i$ for all $a_i\in[0,1]$;
    \item[(ii)] if %
    such equilibrium strategies exist, it is unique and admits the following form:
    \begin{equation}
b(a_i\mid n_2) = 
   g^{-1}\left(   \int_{0}^{a_i}v(x\mid n_2)\sum\limits_{\ell=1}^{n_2-1}(V_\ell-V_{\ell+1})dF_{(\ell,n_1-1)}(x)\right),
\label{eq_b(a_i)_general_prize}
\end{equation}  
for any admitted player with ability $a_i\in[0,1]$.
\end{enumerate}
\end{theorem}

\Cref{thm_equilibrium strategy_general_prize}~(i) shows that no symmetric and strictly increasing (SSI) equilibrium strategy exists when there exists some intervals in $[0,1]$ where  the inflated ability $v(a_i\mid n_2)$ is strictly decreasing in $a_i$. This finding is in sharp contrast with the all-pay auction (contest) literature, where an SSI equilibrium strategy always exists if players are ex ante identical with i.i.d.\ prior distributions; see, e.g., \cite{Moldovanu_2001_AER_WTA_Optimial,Moldovanu_2006_Contest_Architecture_Sub-elimination}. The key difference here is that, in our setting, players have different and private beliefs, so the resulting first-order condition need \textit{not} yield a maximizer.
Further numerical results to illustrate this is provided in \Cref{app_sec:non_increasing}.
We show that it yields a maximizer if and only if the inflated ability $v(a_i\mid n_2)$ is increasing in $a_i$ for all $a_i\in[0,1]$.

We now illustrate why $v(a_i\mid n_2)$ has to be increasing for \textit{all} $a_i\in[0,1]$ to ensure the existence of an SSI equilibrium strategy. When all admitted players follow the strategy in \eqref{eq_b(a_i)_general_prize}, it can be shown that the admitted player $i$'s utility is
\begin{align*}
    \int_0^{a_i} \sum_{\ell=1}^{n_2-1} (V_\ell - V_{\ell+1}) \, F_{(\ell,n_1-1)}(x) \, dv(x\mid n_2).
\end{align*}
When everyone follows a strictly increasing strategy, it can be shown that a higher-ability player receives a higher utility. This requirement underpins why the existence of an SSI equilibrium strategy demands that $v^\prime(x\mid n_2)\geq 0$ for all $x\in[0,1]$, i.e., $v(x\mid n_2)$ is increasing for \textit{all} $x\in[0,1]$. A formal proof is provided in the Appendix.

The potential non-existence of an SSI equilibrium in all-pay contests (auctions) is not uncommon when players' private types are correlated. For instance, \cite{krishna_1997_all_pay_affiliation} establish a sufficient condition for the existence of an SSI equilibrium: namely, when private types are not too strongly \textit{affiliated}, where \textit{affiliation} is defined in the sense of \cite{milgrom_1982_auctiontheory_competitive_bidding}.\footnote{Private types are \textit{affiliated} if their joint density is log-supermodular.} However, to the best of our knowledge, a (simple) necessary condition for the existence of an SSI equilibrium strategy is relatively rare in the all-pay contests (auction) literature.

\Cref{thm_equilibrium strategy_general_prize}~(ii) states that, when an SSI equilibrium strategy exists, it is unique and is given by \eqref{eq_b(a_i)_general_prize}. The uniqueness stems from the fact that the first-order condition admits a unique solution. When $n_2=n_1$, i.e., a regular one-round contest without elimination, $v(x\mid n_2)=x$ and equation \eqref{eq_b(a_i)_general_prize} degenerates to the setting studied by \cite{Moldovanu_2001_AER_WTA_Optimial}. %
It can thus be seen that all effects of the elimination mechanisms to the equilibrium strategy, including the impact of the number of admitted players and the posterior beliefs, are encapsulated in the inflated ability $v(x\mid n_2)$.

Combining \eqref{eq_b(a_i)_general_prize} and \Cref{lemma_J_increasing_n_2}, we have the following corollary.

\begin{corollary}
\label{coro:comparison_n_2_n_2_1}
For any $n_2 \in [3,n_1]$, if $V_{n_2}=0$, $b(a_i\mid n_2-1)\geq b(a_i\mid n_2)$ for all $a_i\in[0,1]$.  
\end{corollary}

The assumption $V_{n_2}=0$ is to make the comparison between $b(a_i\mid n_2-1)$ and $b(a_i\mid n_2)$ fair in the sense that the latter case does not have one additional prize $V_{n_2}$.
The decreasing property follows from the fact that the inflated ability decreases with $n_2$ by \Cref{lemma_J_increasing_n_2}. Note that this decreasing property holds for each realized ability $a_i$, making it a \textit{pointwise} result. To our knowledge, such a pointwise monotonicity result is uncommon in the contest literature.
For example, even in a regular one-round contest without elimination, there is no pointwise monotonicity of players' equilibrium efforts with respect to the total number.\footnote{Consider a situation with a winner-take-all prize structure and a uniform prior distribution.
Then, for a regular one-round contest, the equilibrium effort is $b(a_i) = (1-\frac{1}{n})\cdot a_i^n$, where $n$ is the total number.
For $a_i=0.8$, her effort under $n=3$ is strictly greater than the one under $n=2$ and $n=4$.
}

\subsection{Optimal Admitted Number}

To derive the optimal $n_2$, in what follows, we focus on parameters ($n_1$ and $F$) under which there exists an SSI equilibrium strategy for all $n_2\in[2,n_1]$. 
For example, under the power-law prior distribution, an equilibrium exists if $(n_1-2)\theta \leq 1$ by \Cref{prop:monotonicity of v}.

\begin{theorem}%
\label{thm:opt_number}
 Suppose that
\begin{enumerate}
    \item[(i)]  the inflated ability $v(a_i\mid n_2)$ is non-decreasing in $a_i\in[0,1]$ for any $n_2\in[2,n_1]$;
     \item[(ii)] the designer offers at most two strictly positive prizes. %
\end{enumerate} 
Then, the optimal number of admitted players is two, and the optimal prize structure is winner-take-all in terms of the expected highest effort. 
\end{theorem}

In fact, in the proof of Theorem \ref{thm:opt_number}, we show an even stronger result that among the class of prize structures with at most two strictly positive prizes: (i) given any admitted number, the optimal prize structure is winner-take-all; and (ii) given any prize structure, the optimal admitted number is two. 

When there are more than three positive prizes, one cannot directly compare $b(a_i\mid n_2)$ and $b(a_i\mid n'_2)$ using \eqref{eq_b(a_i)_general_prize} for arbitrary $n_2$ and $n'_2$. In this case, to find the optimal prize structure in terms of the expected highest effort, one needs to explicitly compute $\int_0^1 b(a_i\mid n_2)\,dF^{n_1}(a_i)$.
Because of the complicated form of $v(a_i\mid n_2)$, the expected highest effort is challenging to calculate, even under a linear cost function. Finding the jointly optimal admitted number and prize structure thus calls for future research.

\section{Extension: A Two-Stage Model}
\label{sec_two-stage sec}

When the designer does not have access to the exact ranking of players' abilities, he could run some preliminary tests in the first stage of the contest to learn such information and then decide who to admit.
In this section, we consider a natural extension to our model: a two-stage sequential elimination contest (SEC). In the first stage, all $n_1$ players are able to attend a preliminary contest, and depending on their first-stage efforts, the designer admits the top $n_2$ players to the second stage to compete for prizes. Similar to before, the admitted players only know they and their opponents are the players with top $n_2$ first-stage efforts, but do not know the exact ranking. 
We allow cost functions to be different at different stages. 

To ease the presentation, we delegate the formal definition of the two-stage SEC to Appendix \ref{sec_appendix_sec}. Instead, we provide an informal definition of our solution concept of the two-stage SEC. Our equilibrium notion is based on perfect Bayesian equilibrium (PBE).
\begin{definition}[Informal]
A PBE of a two-stage SEC is a tuple of players' posterior beliefs and players' efforts in two stages that satisfies:
\begin{enumerate}
    \item[(i)] \emph{Bayesian Updating}: Any  player's posterior (second-stage) belief should be updated by the signal that her first-stage effort is among the top $n_2$ highest first-stage efforts.
    \item[(ii)] \emph{Sequential Rationality}: Any player exerts effort that maximizes the expected utility from any stage onward given her belief at the current stage.
\end{enumerate}

\end{definition}
The formal definition can be found in Definition \ref{def_PBE in SEC} in the Appendix. We call a PBE symmetric and strictly increasing if at each stage, all players follow the same strictly increasing %
function which maps their abilities to efforts. The first-stage effort function can be different from the second-stage one. 

\smallskip

Surprisingly, we establish a strong negative result. We show that there does \emph{not} exist a symmetric and strict increasing PBE under any ranking-based prize structure, type distribution, and cost function. We provide a proof sketch below and the formal proof can be found in Appendix \ref{subsec:non_existence}. %

\begin{proposition}[Non-Existence]
\label{prop_non-existenc_sec}
For any $n_2\in[2,n_1)$, there does not exist a symmetric and strictly increasing PBE in a two-stage SEC.
\end{proposition}
\begin{proof}[Proof Sketch]
First, if the inflated ability \(v_i(a_i \mid n_2)\) is not increasing in \(a_i\) for all \(a_i \in [0,1]\), then by \Cref{thm_equilibrium strategy_general_prize}, no symmetric and strictly increasing equilibrium strategy exists in the second stage, which completes the proof for this case.

When the inflated ability $v_i(a_i\mid n_2)$ is indeed always increasing in $a_i$ for all $a_i\in[0,1]$, our proof involves three steps. (i) We first assume that there exists such a symmetric and strictly increasing equilibrium strategy in each stage, at the equilibrium path, the ranking of first-stage efforts has to be exactly the ranking of players' abilities. In this case, the posterior beliefs will be the same as those characterized in Proposition \ref{prop_post_belief}. (ii) Based on the assumption of symmetric and strictly increasing PBE, the posterior beliefs only depend on players' abilities and are independent of first-stage efforts. Thus, backward induction can be used to derive the equilibrium strategies at both stages. We show that the equilibrium efforts at each stage are unique. (iii) Finally, we establish a contradiction by proving that, given that other players follow the above unique equilibrium strategies, there is \emph{always} an incentive for the player to deviate from her first-stage effort and change her posterior belief, which violates the definition of PBE.
\end{proof}

The difficulty to construct a symmetric and strictly increasing PBE lies in the fact that if a player deviates from her first-stage strategy, her second-stage beliefs will change accordingly. This dependency makes it complicated to derive an equilibrium in a two-stage SEC.

\smallskip

We now contrast our analyses and results of the two-stage SEC with other existing results. \citet{Moldovanu_2006_Contest_Architecture_Sub-elimination} studied another two-stage elimination contest format, wherein in the first stage, all players are randomly divided into $n_2$ groups, and the winner in each sub-group, i.e., the player with the highest first-stage effort within each sub-group, is promoted to the second and final stage. We refer to this format as a two-stage \textit{sub-elimination} contest. They show that there exists a unique symmetric and strictly increasing PBE strategy. When $n_1$ is large enough, regardless of the prior distribution, it is optimal to set $n_2=2$ in terms of the expected highest efforts under the linear cost function. They further show that this outcome is strictly better than that under a regular one-round contest admitting all players.

The analysis of sub-elimination contests is much easier than that of SEC since players' posterior beliefs are \textit{independent} of first-stage efforts. This is because the players in the second stage are in different sub-groups in the first stage, thus conditional on being promoted to the second stage, the first-stage efforts 
have no impact on the posterior beliefs about her opponents' abilities. Specifically, in the sub-elimination contest, players in the second stage will perceive their opponents' abilities as the largest order statistic among $n_1/n_2$ i.i.d. random variables drawn from the prior distribution. However, in a two-stage SEC, players' posterior beliefs are coupled with their first-stage efforts. %

\section{Closing Remarks}
\label{sec_closing}

We characterize the players' private posterior beliefs and show that all effects of this elimination mechanism are captured by their inflated abilities. Furthermore, the existence of a symmetric and strictly increasing equilibrium strategy is equivalent to the inflated ability being increasing in each player's true ability. When the designer offers at most two positive prizes, we show that admitting two players and using a winner-take-all prize structure is optimal in terms of the expected highest efforts.

In the extended two-stage model, although we show that no symmetric and strictly increasing PBE exists, we conjecture that an asymmetric PBE may exist. Specifically, players' first-stage equilibrium efforts remain symmetric, but their second-stage equilibrium efforts may be asymmetric due to non-identical posterior beliefs. This conjecture calls for future analysis.

\bibliographystyle{ACM-Reference-Format}
\bibliography{references}

\newpage
\appendix

\begin{center}
\textbf{\LARGE Appendix}
\end{center}
\medskip

In \Cref{app_sec:non_increasing}, we provide numerical results when the inflated ability is not increasing with the true ability.
In \Cref{sec_appendix:numerical}, we provide additional results about posterior beliefs.
In Appendix \ref{sec_appendix_pf}, we provide detailed proofs of the main results.
In Appendix \ref{sec_appendix_auxiliaryreults}, we give additional auxiliary lemmas and their proofs.
Finally, in Appendix \ref{sec_appendix_sec}, we present the formal model and analysis of a two-stage sequential elimination contest.

\section{When the Inflated Ability is Not Increasing}
\label{app_sec:non_increasing}

\Cref{Fig:utility_envelopefails} shows how players' utilities behave when the inflated ability $v(a_i\mid n_2)$ is not always increasing in $a_i\in[0,1]$. Specifically, it plots how $U_i(\tilde{a}_i,a_i \mid n_2,b)$, which is the admitted player $i$'s utility when all other players adopt strategy $b(\cdot)$ given by \eqref{eq_b(a_i)_general_prize} while she exerts effort $b(\tilde{a}_i)$, varies with $\tilde{a}_i$. If $b(\cdot)$ is an equilibrium, then $U_i(\tilde{a}_i,a_i \mid n_2,b)$ should be maximized at $\tilde{a}_i=a_i$. However, we see that when the inflated ability is not increasing, this does not hold.
\begin{figure*}[htbp]
  \centering
  \subfloat[$a_i=0.3$]{  \resizebox{.32\textwidth}{!}{\begin{tikzpicture}
\begin{axis}[
    xlabel={$\tilde{a}_i$},
    ylabel={$U_i(\tilde{a}_i,a_i \mid n_2,b)$},
    thick,
    no markers,
    every axis plot post/.append style={line width=1.8pt}, %
    tick style={line width=1pt, black}, %
    tick label style={font=\large, scale=1.2}, %
    label style={font=\large, scale=1.2}
]

\addplot[blue] table[x index=0, y index=1] {fig/envelopefails/data/utility_foc_a_i=0.30.txt};

\draw[dotted, line width=1.5pt] 
  (axis cs:0.3,0) 
  -- (axis cs:0.3,-0.9);

\end{axis}
\end{tikzpicture}}}
  \hfill
  \subfloat[$a_i=0.5$]{  \resizebox{.32\textwidth}{!}{\begin{tikzpicture}
\begin{axis}[
    xlabel={$\tilde{a}_i$},
    ylabel={$U_i(\tilde{a}_i,a_i \mid n_2,b)$},
    thick,
    no markers,
    every axis plot post/.append style={line width=1.8pt}, %
    tick style={line width=1pt, black}, %
    tick label style={font=\large, scale=1.2}, %
    label style={font=\large, scale=1.2}
]

\addplot[blue] table[x index=0, y index=1] {fig/envelopefails/data/utility_foc_a_i=0.50.txt};

\draw[dotted, line width=1.5pt] 
  (axis cs:0.5,-0.01) 
  -- (axis cs:0.5,-0.65);

\end{axis}
\end{tikzpicture}}}
  \hfill
  \subfloat[$a_i=0.7$]{  \resizebox{.32\textwidth}{!}{\begin{tikzpicture}
\begin{axis}[
    xlabel={$\tilde{a}_i$},
    ylabel={$U_i(\tilde{a}_i,a_i \mid n_2,b)$},
    thick,
    no markers,
    every axis plot post/.append style={line width=1.8pt}, %
    tick style={line width=1pt, black}, %
    tick label style={font=\large, scale=1.2}, %
    label style={font=\large, scale=1.2}
]

\addplot[blue] table[x index=0, y index=1] {fig/envelopefails/data/utility_foc_a_i=0.70.txt};

\draw[dotted, line width=1.5pt] 
  (axis cs:0.7,-0.05) 
  -- (axis cs:0.7,-0.4);

\end{axis}
\end{tikzpicture}}}
  \vspace{0.9em}
  \subfloat[$a_i=0.8$]{  \resizebox{.32\textwidth}{!}{\begin{tikzpicture}
\begin{axis}[
    xlabel={$\tilde{a}_i$},
    ylabel={$U_i(\tilde{a}_i,a_i \mid n_2,b)$},
    thick,
    no markers,
    every axis plot post/.append style={line width=1.8pt}, 
    tick style={line width=1pt, black}, 
    tick label style={font=\large, scale=1.2}, 
    label style={font=\large, scale=1.2},
    ymin=-0.3,
    ytick={-0.3, 0}
]

\addplot[blue] table[x index=0, y index=1] 
  {fig/envelopefails/data/utility_foc_a_i=0.80.txt};

\draw[dotted, line width=1.5pt] 
  (axis cs:0.8,-0.08) 
  -- (axis cs:0.8,-0.3);

\end{axis}
\end{tikzpicture}}}
  \hfill
  \subfloat[$a_i=0.9$]{  \resizebox{.32\textwidth}{!}{\begin{tikzpicture}
\begin{axis}[
    xlabel={$\tilde{a}_i$},
    ylabel={$U_i(\tilde{a}_i,a_i \mid n_2,b)$},
    thick,
    no markers,
    every axis plot post/.append style={line width=1.8pt}, 
    tick style={line width=1pt, black}, 
    tick label style={font=\large, scale=1.2}, 
    label style={font=\large, scale=1.2},
    ymin=-0.2,
    ytick={-0.2,  0}
]

\addplot[blue] table[x index=0, y index=1] {fig/envelopefails/data/utility_foc_a_i=0.90.txt};

\draw[dotted, line width=1.5pt] 
  (axis cs:0.9,-0.11) 
  -- (axis cs:0.9,-0.2);

\end{axis}
\end{tikzpicture}}}
  \hfill
  \subfloat[$a_i=0.98$]{  \resizebox{.32\textwidth}{!}{\begin{tikzpicture}
\begin{axis}[
    xlabel={$\tilde{a}_i$},
    ylabel={$U_i(\tilde{a}_i,a_i \mid n_2,b)$},
    thick,
    no markers,
    every axis plot post/.append style={line width=1.8pt}, %
    tick style={line width=1pt, black}, %
    tick label style={font=\large, scale=1.2}, %
    label style={font=\large, scale=1.2},
    ytick={-0.1,0}  %
]

\addplot[blue] table[x index=0, y index=1] {fig/envelopefails/data/utility_foc_a_i=0.98.txt};

\draw[dotted, line width=1.5pt] 
  (axis cs:0.98,-0.13) 
  -- (axis cs:0.98,-0.085);

\end{axis}
\end{tikzpicture}}}
  \caption{Player $i$'s utility $U_i(\tilde{a}_i,a_i \mid n_2,b)$ ($n_1=10$, $n_2=3$) under a uniform prior and a linear cost function with a winner-take-all prize structure.
  The vertical axis $U_i(\tilde{a}_i,a_i \mid n_2,b)$ represents the utility of an admitted player $i$ with ability $a_i$ if all other admitted players use the strategy $b(\cdot)$ defined in \eqref{eq_b(a_i)_general_prize} while she exerts effort $b(\tilde{a}_i)$.
}
\label{Fig:utility_envelopefails}
\end{figure*}

\section{Non-differentiable Posterior Beliefs}
\label{sec_appendix:numerical}
Recall that admitted player $i$'s marginal posterior belief (density) is non-differentiable at $a_j=a_i$. Define the \textit{jump} in the derivative of posterior PDFs as follows.
\begin{align*}
    H(a_i,n_2\mid F,n_1):=\lim_{a_j\to a_i^{-}}\frac{\partial\beta_i\left(a_j\mid s,a_i \right)}{\partial a_j}-\lim_{a_j\to a_i^{+}}\frac{\partial\beta_i\left(a_j\mid s,a_i \right)}{\partial a_j}.
\end{align*}

\begin{figure*}[htbp]
  \centering
  \resizebox{.68\textwidth}{!}{\begin{tikzpicture}
\begin{axis}[
    xlabel={$a_i$},
    ylabel={$n_2$},
    zlabel={$H(a_i, n_2 \mid F, n_1)$},
    view={60}{30},
    grid=both,
    mesh/cols=100,          %
    mesh/rows=19,           %
    mesh/ordering=colwise,  %
    unbounded coords=jump,
    ytick={2,4,6,8,10,12,14,16,18,20}, %
    tick label style={font=\footnotesize	}, %
    label style={font=\footnotesize	},      %
]

\addplot3[
    mesh,
    draw=gray,             %
    line width=0.2pt,     %
]
table[
    x=ai,
    y=n2,
    z=H,
    col sep=tab,
    header=true
]
{fig/belief/data/H_values_data_uniformprior.txt};

\end{axis}
\end{tikzpicture}}
  \caption{The jump $H(a_i,n_2\mid F,n_1)$ ($n_1=20$ and uniform prior distribution)}
\label{fig_jump}
\end{figure*}

Figure \ref{fig_jump} shows how the jump varies with $a_i$ and $n_2$ given uniform prior distribution and $n_1=20$. 
The amount of this jump can measure the impact of the admission signal on players' posterior beliefs. As we can see in Figure \ref{fig_jump}, in general, the amount of jump is not monotone in $a_i$. However, numerical results suggest that it is non-increasing in $n_2$.
Our next proposition proves this monotonicity for uniform distribution and $a_i\geq 0.5$.

\begin{proposition}
For uniform prior distribution, given $a_i\in[0.5,1]$, $H(a_i,n_2\mid F,n_1)$ is non-increasing with $n_2\in[2,n_1]$ for any $n_1$. 
\end{proposition}

\begin{proof}
For \( a_j < a_i \), we have
\begin{align*}
\frac{\partial \beta_i(a_j \mid s, a_i)}{\partial a_j} &= \frac{1}{I(a_i, n_2)} \left[ (n_2 - 2) \frac{\partial}{\partial a_j} B\left( a_j, n_1 - n_2 + 1,\ n_2 - 2 \right) + \frac{\partial}{\partial a_j} \left( a_j^{n_1 - n_2} (1 - a_j)^{n_2 - 2} \right) \right]\\
&= \frac{1}{I(a_i, n_2)} \left[ (n_2 - 2) a_j^{n_1 - n_2} (1 - a_j)^{n_2 - 3} \right. + (n_1 - n_2) a_j^{n_1 - n_2 - 1} (1 - a_j)^{n_2 - 2} \\
&\quad \left. - (n_2 - 2) a_j^{n_1 - n_2} (1 - a_j)^{n_2 - 3} \right] \\
&= \frac{(n_1 - n_2) a_j^{n_1 - n_2 - 1} (1 - a_j)^{n_2 - 2}}{I(a_i, n_2)},
\end{align*}
hence
\[
\lim_{a_j \to a_i^-} \frac{\partial \beta_i(a_j \mid s, a_i)}{\partial a_j} = \frac{(n_1 - n_2) a_i^{n_1 - n_2 - 1} (1 - a_i)^{n_2 - 2}}{I(a_i, n_2)}.
\]
Under the uniform distribution, we obtain 
\[
\lim_{a_j \to a_i^+} \frac{\partial \beta_i(a_j \mid s, a_i)}{\partial a_j} = 0.
\]
Therefore, the jump is
\[
H(a_i, n_2\mid F,n_1) = \frac{(n_1 - n_2) a_i^{n_1 - n_2 - 1} (1 - a_i)^{n_2 - 2}}{I(a_i, n_2)}.
\]
Since $I(a_i,n_2)$ is increasing with $n_2$, and observe that 
\begin{align*}
    \frac{H(a_i, n_2 + 1\mid F,n_1)}{H(a_i, n_2\mid F,n_1)} = \frac{n_1 - n_2 -1}{n_1 -n_2}\frac{I(a_i, n_2 )}{I(a_i, n_2 + 1)}\frac{1 - a_i}{a_i},
\end{align*}
we can conclude that  $H(a_i, n_2 + 1\mid F,n_1)/H(a_i, n_2\mid F,n_1)\leq 1$ when $a_i \in [0.5, 1]$, which completes the proof.
\end{proof}

\section{Proofs}
\label{sec_appendix_pf}
\begin{proof}[\textbf{Proof of Proposition \ref{prop_post_belief}}]
Denote by $\mathcal{I}^0$ the set of $n_1$ initial players. For any admitted player $i \in \mathcal{I}$, define $\mathcal{I}^0_{-i} = \{j: j\in \mathcal{I}^0, j\neq i\}$, and $\mathcal{I}_{-i} = \{j: j\in \mathcal{I}, j\neq i\}$. With a little notation abuse, let $a^0_{-i}=\left[a_j : j\in\mathcal{I}^0_{-i} \right]$,
$a_{(1)} = \min_{j \in \mathcal{I}_{-i}} a_j$, and $a^*=\max_{j \in \mathcal{I}^0\backslash\mathcal{I}} a_j$. 

Given signal $s$, by the Bayes' theorem and the tower property of conditional expectation, we have 
 \begin{align*}
    \beta_i\left(a^0_{-i} \mid s,n_2, a_i\right)& = \beta_i\left(a^0_{-i} \mid s,a_i\right) = \beta_i\left(a^0_{-i} \mid s_{-i},s_{i},a_i\right)\\
    & = \frac{\beta_i\left(a^0_{-i} \mid s_{i}, a_i\right)\Pr\left(s_{-i}\mid a^0_{-i},s_{i}, a_i\right)}{\int_{[0,1]^{n_1-1}}\beta_i\left(a^0_{-i} \mid s_{i}, a_i\right)\Pr\left(s_{-i}\mid a^0_{-i},s_{i}, a_i\right)da^0_{-i}} \nonumber,
\end{align*}
where $\beta_i\left(a^0_{-i} \mid s,a_i\right)$ is the belief of the other $n_1-1$ players' abilities on the condition that player $i$ is admitted and her ability is $a_i$, and $\Pr\left(s_{-i}\mid a^0_{-i},s_{i}, a_i\right)$ is the probability that players in $\mathcal{I}_{-i}$ get admitted into the contest on the condition that player $i$ is admitted, her ability is $a_i$, and all the other players' abilities $a^0_{-i}$.
Observe that the formula $\beta_i\left(a^0_{-i} \mid s_i,a_i\right)$ does not condition on $s_{-i}$, then the admitted player $i$ can not obtain additional information of $n_2$ and the other players' abilities, hence  $\beta_i\left(a^0_{-i} \mid s_i,a_i\right)$ is exactly the prior $\prod_{j\in\mathcal I^0_{-i}}f(a_j)$.

Since player $i$ is admitted, if the minimum value of other players' abilities $a_{(1)}$ is greater than player $i$'s ability, or $a_{(1)}$ is less than player $i$'s ability but greater than the maximum value of unadmitted players' abilities $a^*$, $\Pr\left(s_{-i}\mid a^0_{-i},s_{i}, a_i\right)=1$; otherwise, $\Pr\left(s_{-i}\mid a^0_{-i},s_{i}, a_i\right)=0$. Formally, we have $\Pr\left(s_{-i}\mid a^0_{-i},s_{i}, a_i\right)=\mathbb{I}\left\{ a^*<a_i\leq a_{(1)}\right\}+\mathbb{I}\left\{a^*<a_{(1)}< a_i\right\}$, where $\mathbb{I}\{\cdot\}$ is the indicator function. The denominator in the posterior beliefs then becomes
\begin{align}
&\quad \int_{[0,1]^{n_1-1}} \beta_i(a^0_{-i}\mid s_i, a_i) \Pr\left(s_{-i}\mid a_{-i}^{0},s_i, a_i\right)   da_{-i}^0\nonumber \\
&= \underbrace{\int_0^1\cdots\int_0^1}_{n_1-1} \Pr\left(s_{-i}\mid a^0_{-i},s_i, a_i\right) \prod_{j\in \mathcal{I}^0_{-i}}\left(f(a_j)da_j \right)\nonumber \\
& = (n_2-1)\bigg( \int_0^{a_i} \biggl(\underbrace{\int_{a_k}^1\cdots \int_{a_k}^1}_{n_2-2}\biggl(\underbrace{\int_{0}^{a_k}\cdots \int_{0}^{a_k}}_{n_1-n_2}\prod_{j \in \mathcal{I}^0\backslash \mathcal{I}}\bigl(f(a_j)da_j \bigl)\biggl)  \prod_{j \in \mathcal{I}_{-i}\backslash\{k\}}\bigl(f(a_j)da_j \bigl) \biggl)f(a_k)\underbrace{da_k}_{k \in \mathcal{I}_{-i}}  \biggl) \nonumber \\
& \quad + (n_2-1)\biggl( \int^1_{a_i} \biggl(\underbrace{\int_{a_k}^1\cdots \int_{a_k}^1}_{n_2-2} \biggl(\underbrace{\int_{0}^{a_i}\cdots \int_{0}^{a_i}}_{n_1-n_2}\prod_{j \in \mathcal{I}^0\backslash \mathcal{I}}\bigl(f(a_j)da_j \bigl)\biggl)  \prod_{j \in \mathcal{I}_{-i}\backslash\{k\}}\bigl(f(a_j)da_j \bigl) \biggl)f(a_k)\underbrace{da_k}_{k \in \mathcal{I}_{-i}}   \biggl) \nonumber \\
& = (n_2-1)\int_0^{a_i} (1-F(a_k) \bigl)^{n_2-2} F^{n_1-n_2}(a_k)dF(a_k) + (n_2-1)F^{n_1-n_2}(a_i)\int^1_{a_i} \bigl(1-F(a_k) )^{n_2-2} dF(a_k) \nonumber \\
& = (n_2-1)B\left(F(a_i),n_1-n_2+1,n_2-1)\right) + F^{n_1-n_2}(a_i)\bigl(1 - F(a_i) \bigl)^{n_2-1}:=I(F(a_i),n_2). \nonumber 
\end{align}
Hence   
\begin{align*}
    \beta_i\left(a^0_{-i} \mid s,a_i\right)= \frac{\mathbb{I}\left\{ a^*<a_i\leq a_{(1)}\right\}+\mathbb{I}\left\{a^*<a_{(1)}< a_i\right\}}{I(F(a_i),n_2)} \prod\limits_{j \in \mathcal{I}^0_{-i}}f(a_j)\nonumber,
\end{align*}
which implies that
\begin{align*}
\beta_i(a_{-i} \mid s ,a_i) =& \underbrace{\int_{0}^{1}\cdots \int_{0}^{1}}_{n_1-n_2}\beta_i\left(a^0_{-i} \mid s,a_i\right)\prod_{j \in \mathcal{I}^0\backslash \mathcal{I}}da_j \\
=&\begin{cases}
\frac{F^{n_1-n_2}\left(a_i\right)}{I(F(a_i),n_2)}\prod\limits_{j \in \mathcal{I}_{-i}}f(a_j), & a_i \leq \normalfont{\min}_{j\in \mathcal{I}_{-i}} a_j , \\[4mm]
\frac{F^{n_1-n_2}\left(\normalfont{\min}_{j\in \mathcal{I}_{-i}} a_j\right)}{I(F(a_i),n_2)}\prod\limits_{j \in \mathcal{I}_{-i}}f(a_j), ~ & a_i > \normalfont{\min}_{j\in \mathcal{I}_{-i}} a_j.
\end{cases}
\end{align*}
This completes the proof.
\end{proof}

\medskip

\begin{proof}[\textbf{Proof of Corollary \ref{Marginal Posterior Belief}}]
 For any admitted players $i\neq j \in \mathcal{I}$, define $\mathcal{I}_{-i,-j} = \mathcal{I}\backslash\{i,j\}$. With a little notation abuse, denote ${\tilde a_{(1)}} = \min_{k \in \mathcal{I}_{-i,-j}} a_k $ as the minimum ability of admitted players except for player $i$ and player $j$.

\smallskip

When $a_i<a_j<1$, by Proposition \ref{prop_post_belief}, the (joint) posterior beliefs is given as follows.
\[
\beta_{i}(a_{-i}\mid s,a_i) = 
\begin{cases}
\frac{F^{n_1-n_2}(a_i)}{I(F(a_i),n_2)}\prod_{k\in \mathcal{I}_{-i}}f(a_k), \qquad & a_i < {\tilde a_{(1)}}<1,  \vspace{0.5cm}\\
\frac{F^{n_1-n_2}\left({\tilde a_{(1)}}\right)}{I(F(a_i),n_2)}\prod_{k\in \mathcal{I}_{-i}}f(a_k), \qquad & 0<{\tilde a_{(1)}}<a_i . 
\end{cases}
\]
Hence the marginal posterior belief becomes
\begin{align}
\beta_{i}(a_{j}\mid s,a_i)\nonumber =& \frac{(n_2-2)}{I(F(a_i),n_2)}\bigg( \int_0^{a_i} \biggl(\underbrace{\int_{a_\ell}^1\cdots \int_{a_\ell}^1}_{n_2-3} f(a_j)f(a_\ell) F^{n_1-n_2}(a_\ell) \prod_{k\neq \ell, k \in \mathcal{I}_{-i,-j}}\bigl(f(a_k)da_k \bigl) \biggl)\underbrace{da_\ell}_{\ell \in \mathcal{I}_{-i,-k}}  \bigg) \nonumber \\ 
&+\frac{(n_2-2)F^{n_1-n_2}(a_i)}{I(F(a_i),n_2)}\bigg( \int^1_{a_i} \biggl(\underbrace{\int_{a_\ell}^1\cdots \int_{a_\ell}^1}_{n_2-3} f(a_j)f(a_\ell)  \prod_{k\neq \ell, k \in \mathcal{I}_{-i,-j}}\bigl(f(a_k)da_k \bigl) \biggl)\underbrace{da_\ell}_{\ell \in \mathcal{I}_{-i,-k}}   \bigg) \nonumber \\
=& \frac{(n_2-2)f(a_j)}{I(F(a_i),n_2)} \bigg(\int_0^{a_i} \bigl(1-F(a_\ell) \bigl)^{n_2-3} F^{n_1-n_2}(a_\ell)dF(a_\ell)+F^{n_1-n_2}(a_i)\int^1_{a_i} \bigl(1-F(a_\ell) \bigl)^{n_2-3} dF(a_\ell)\bigg) \nonumber \\
=& \frac{f(a_j)}{I(F(a_i),n_2)}\left((n_2-2)B\left(F(a_i),n_1-n_2+1,n_2-2)\right) + F^{n_1-n_2}(a_i)\bigl(1 - F(a_i) \bigl)^{n_2-2} \right).\nonumber
\end{align} 

Similarly, the (joint) posterior beliefs when $0<a_j<a_i$ is 
\[
\beta_{i}(a_{-i}\mid s,a_i) = 
\begin{cases}
\frac{F^{n_1-n_2}(a_j)}{I(F(a_i),n_2)}\prod_{k\in \mathcal{I}_{-i}}f(a_k), \qquad & a_j < {\tilde a_{(1)}}<1,  \vspace{0.5cm}\\
\frac{F^{n_1-n_2}\left({\tilde a_{(1)}}\right)}{I(F(a_i),n_2)}\prod_{k\in \mathcal{I}_{-i}}f(a_k), \qquad & 0<{\tilde a_{(1)}}<a_j . 
\end{cases}
\] 
Thus, the marginal posterior belief in this case becomes
\begin{align}
\beta_{i}(a_{j}\mid s,a_i) 
& = \frac{(n_2-2)}{I(F(a_i),n_2)}\biggl( \int_0^{a_j} \biggl(\underbrace{\int_{a_\ell}^1\cdots \int_{a_\ell}^1}_{n_2-3} f(a_j)f(a_\ell) F^{n_1-n_2}(a_\ell) \prod_{k\neq \ell, k \in \mathcal{I}_{-i,-j}}\bigl(f(a_k)da_k \bigl) \biggl)\underbrace{da_\ell}_{\ell \in \mathcal{I}_{-i,-k}}  \biggl)  \nonumber \\ 
& \quad +\frac{(n_2-2)}{I(F(a_i),n_2)}\biggl( \int^1_{a_j} \biggl(\underbrace{\int_{a_\ell}^1\cdots \int_{a_\ell}^1}_{n_2-3} f(a_j)f(a_\ell)F^{n_1-n_2}(a_j) \prod_{k\neq \ell, k \in \mathcal{I}_{-i,-j}}\bigl(f(a_k)da_k \bigl) \biggl)\underbrace{da_\ell}_{\ell \in \mathcal{I}_{-i,-k}}   \biggl) \nonumber \\
& = \frac{(n_2-2)}{I(F(a_i),n_2)}f(a_j)\int_0^{a_j} \bigl(1-F(a_\ell) \bigl)^{n_2-3} F^{n_1-n_2}(a_\ell)f(a_\ell)da_\ell \nonumber \\ 
&\quad + \frac{(n_2-2)}{I(F(a_i),n_2)}f(a_j)F^{n_1-n_2}(a_j)\int^1_{a_j} \bigl(1-F(a_\ell) \bigl)^{n_2-3} f(a_\ell)da_\ell \nonumber \\
& = \frac{f(a_j)}{I(F(a_i),n_2)}\left((n_2-2)B(F(a_j),n_1-n_2+1,n_2-2)+F^{n_1-n_2}(a_j)(1-F(a_j))^{n_2-2}\right).\nonumber 
\end{align} 
Putting everything together, the marginal posterior belief is given by
\begin{align}\label{marginal belief}
&\quad \beta_i\left(a_j\mid s,a_i \right) \nonumber\\
&= \begin{cases}
\frac{f(a_j)}{I(F(a_i),n_2)}\left((n_2-2)B(F(a_j),n_1-n_2+1,n_2-2)+F^{n_1-n_2}(a_j)\bigl(1-F(a_j)\bigl)^{n_2-2}\right), & a_j < a_i,\\
\frac{f(a_j)}{I(F(a_i),n_2)}\left((n_2-2)B\left(F(a_i),n_1-n_2+1,n_2-2\right) +F^{n_1-n_2}(a_i) \bigl(1 - F(a_i) \bigl)^{n_2-2} \right), & a_j > a_i.\\
\end{cases}
\end{align}
This completes the proof.
\end{proof}

\medskip
\begin{proof}[\textbf{Proof of Proposition \ref{Proposition_margin_distribution_Stochastic Dominance}}]
It is equivalent to show: if $2 \leq n_2 < n_1$, $\forall a_j \in [0,1]$,
\[
    \Pr_{A_{j}\sim \beta_i(\cdot\mid s,a_i)}\left(A_{j} \leq a_j \right) \leq \Pr_{A_{j} \sim f(\cdot)}\left(A_{j} \leq a_j\right),
\]
where the right-hand-side is
\[
\Pr_{A_{j} \sim f(\cdot)}\left(A_{j} \leq a_j\right)=F(a_j).
\]
At the beginning of the proof, we show that the following two equations hold.

\medskip

\noindent(i). For all $\alpha\in \mathbb{N}_{>0}\setminus\{1,2\},  \beta \in \mathbb{N}_{>0}\setminus\{1\}$,
\begin{align}
    \beta B\left(x,\alpha,\beta \right)&=\beta \int_0^x t^{\alpha-1}(1-t)^{\beta-1}dt= (\alpha-1)\int_0^x t^{\alpha-2}(1-t)^{\beta}dt-x^{\alpha-1}(1-x)^{\beta} .\label{Appendix_Eq_beta}
\end{align}
\noindent(ii). Using (\ref{Appendix_Eq_beta}), we re-calculate $I(F(a_i),n_2)$ as follows.
\begin{align}
    I(F(a_i),n_2)&=(n_2-1)B\left(F(a_i),n_1-n_2+1,n_2-1 \right)+  F^{n_1-n_2}(a_i)(1-F(a_i))^{n_2-1} \nonumber \\
    &=(n_1-n_2)\int_0^{F(a_i)}x^{n_1-n_2-1}(1-x)^{n_2-1}dx \nonumber. 
\end{align}

\noindent\textit{\underline{Case 1: $a_j< a_i$}} 

\vspace{1em}
When $a_j < a_i$, by equation (\ref{marginal belief}) we have
\begin{align}
    &\Pr_{A_{j}\sim \beta_i(\cdot\mid s,a_i)}\left(A_{j} \leq a_j \right) \nonumber \\  
    &=\frac{1}{I(F(a_i),n_2)}\left(\int_0^{F(a_j)}(n_2-2)B\left(s,n_1-n_2+1,n_2-2 \right)ds + \int_0^{F(a_j)}s^{n_1-n_2}(1-s)^{n_2-2}ds \right).\nonumber
\end{align}
Since 
\begin{align}
    &\quad \int_0^{F(a_j)}(n_2-2)B\left(s,n_1-n_2+1,n_2-2 \right)ds\nonumber\\
    &=(n_2-2)\int_0^{F(a_j)}ds\int_0^sx^{n_1-n_2}(1-x)^{n_2-3}dx
    \nonumber\\
    &=(n_2-2)\int_0^{F(a_j)}dx\int_x^{F(a_j)}x^{n_1-n_2}(1-x)^{n_2-3}ds \nonumber\\
    &=(n_2-2)\int_0^{F(a_j)}\left(F(a_j)-x\right)x^{n_1-n_2}(1-x)^{n_2-3}dx \nonumber\\
    &=(n_1-n_2)F(a_j)\int_0^{F(a_j)}x^{n_1-n_2-1}(1-x)^{n_2-2}dx-(n_1-n_2+1)\int_0^{F(a_j)}x^{n_1-n_2}(1-x)^{n_2-2}dx \nonumber,
\end{align}
where the last equality holds by the identity in equation (\ref{Appendix_Eq_beta}), then we have
\begin{align}
    &\Pr_{A_{j}\sim \beta_i(\cdot\mid s,a_i)}\left(A_{j} \leq a_j \right) \nonumber \\  &=\frac{n_1-n_2}{I(F(a_i),n_2)}\left(F(a_j)\int_0^{F(a_j)}x^{n_1-n_2-1}(1-x)^{n_2-2}dx-\int_0^{F(a_j)}x^{n_1-n_2}(1-x)^{n_2-2}dx \right) \label{equa_10}
     \\
    &\leq \frac{F(a_j)(n_1-n_2)}{I(F(a_i),n_2)}\int_0^{F(a_j)}\left(x^{n_1-n_2-1}(1-x)^{n_2-2}-x^{n_1-n_2}(1-x)^{n_2-2}\right)dx  \nonumber \\
    &=\frac{F(a_j)}{I(F(a_i),n_2)}\left((n_1-n_2)\int_0^{F(a_j)}x^{n_1-n_2-1}(1-x)^{n_2-1}dx \right)\nonumber \\
    &=F(a_j) \nonumber \\[2mm]
    &=\Pr_{A_{j} \sim f(\cdot)}\left(A_{j} \leq a_j\right). \nonumber
\end{align}

\bigskip

\noindent\textit{\underline{Case 2: $a_j > a_i$}} 

\vspace{1em}

When $a_j>a_i$, we have
\begin{align}
    &\Pr_{A_{j}\sim \beta_i(\cdot\mid s,a_i)}\left(A_{j} \leq a_j \right) \nonumber \\  &=\frac{1}{I(F(a_i),n_2)}\int_0^{a_i}(n_2-2)B\left(F(t),n_1-n_2+1,n_2-2 \right)dF(t)\nonumber \\
    &\quad+ \frac{1}{I(F(a_i),n_2)} \int_0^{a_i}F^{n_1-n-2}(t)(1-F(t))^{n_2-2}dF(t) \nonumber \\
    &\quad+\frac{1}{I(F(a_i),n_2)}\left(\int_{a_i}^{a_j}(n_2-2)B\left(F(a_i),n_1-n_2+1,n_2-2 \right)+ F^{n_1-n_2}(a_i) (1-F(a_i))^{n_2-2}dF(t) \right) .\nonumber
\end{align}
In addition, by equation (\ref{equa_10}),
\begin{align}
    & \quad \frac{1}{I(F(a_i),n_2)}\left(\int_0^{a_i}\left((n_2-2)B\left(F(t),n_1-n_2+1,n_2-2 \right)+F^{n_1-n-2}(t)(1-F(t))^{n_2-3}\right)dF(t) \right) \nonumber \\
    &=\frac{n_1-n_2}{I(F(a_i),n_2)}\left(F(a_i)\int_0^{F(a_i)}x^{n_1-n_2-1}(1-x)^{n_2-2}dx-\int_0^{F(a_i)}x^{n_1-n_2}(1-x)^{n_2-2}dx \right). \nonumber 
\end{align}
Moreover, 
\begin{align}
    &\quad \frac{1}{I(a_i,n_2)}\left(\int_{a_i}^{a_j}(n_2-2)B\left(F(a_i),n_1-n_2+1,n_2-2 \right)+F^{n_1-n_2}(a_i) (1-F(a_i))^{n_2-2}dF(t) \right) \nonumber \\
    &=\frac{1}{I(a_i,n_2)}(F(a_j)-F(a_i))\left((n_2-2)\int_{0}^{F(a_i)}x^{n_1-n_2}(1-x)^{n_2-3}dx+F^{n_1-n_2}(a_i)(1-F(a_i))^{n_2-2}\right) \nonumber \\
    &=\frac{F(a_j)-F(a_i)}{I(a_i,n_2)}\left((n_1-n_2)\int_{0}^{F(a_i)}x^{n_1-n_2-1}(1-x)^{n_2-2}dx\right) ,\nonumber
\end{align}
where the last equality holds by the identity in equation \eqref{Appendix_Eq_beta}.

\vspace{1em}

Hence, we have
\begin{align}
    &\Pr_{A_{j}\sim \beta_i(\cdot\mid s,a_i)}\left(A_{j} \leq a_j \right) \nonumber \\  &=\frac{n_1-n_2}{I(F(a_i),n_2)}\left(F(a_j)\int_0^{F(a_i)}x^{n_1-n_2-1}(1-x)^{n_2-2}dx-\int_0^{F(a_i)}x^{n_1-n_2}(1-x)^{n_2-2}dx \right) \nonumber \\
    &\leq \frac{n_1-n_2}{I(F(a_i),n_2)}F(a_j)\int_0^{F(a_i)}(x^{n_1-n_2-1}(1-x)^{n_2-2}-x^{n_1-n_2}(1-x)^{n_2-2})dx  \nonumber \\
    &= \frac{F(a_j)}{I(F(a_i),n_2)}\underbrace{\left((n_1-n_2)\int_0^{F(a_i)}x^{n_1-n_2-1}(1-x)^{n_2-1}dx \right)}_{I(F(a_i),n_2)}\nonumber \\
    &=F(a_j) \nonumber \\[3mm]
    &=\Pr_{A_{j} \sim f(\cdot)}\left(A_{j} \leq a_j\right). \nonumber
\end{align}
This completes the proof.
\end{proof}

\medskip

\begin{proof}[\textbf{Proof of Lemma \ref{lemma_J_increasing_n_2}}]
(i). Recall the definition of $J(x,n_2)$,
\begin{align}
     J(x,n_2) & := \binom{n_1-1}{n_2-1} \cdot I(x,n_2) \nonumber \\
& = \binom{n_1-1}{n_2-1} \biggl ((n_2-1)\int_0^{x}t^{n_1-n_2}(1-t)^{n_2-2}dt + x^{n_1-n_2}(1-x)^{n_2-1} \biggl) \nonumber\\
& = \binom{n_1-1}{n_2-1} (n_1-n_2)\int_0^{x}t^{n_1-n_2-1}(1-t)^{n_2-1}dt\nonumber.
\end{align}
Thus, the monotonicity of $J(x,n_2)$ in $x$ is the same as the monotonicity of $I(x,n_2)$ in $x$.
Observe that
\begin{align}
    \frac{\partial I(x,n_2)}{\partial x} &  = (n_1-n_2)x^{n_1-n_2-1}\bigl(1 - x \bigl)^{n_2-1} \geq 0,
\end{align}
then for any $n_2\in [2,n_1]$, $I(x,n_2)$ is increasing in $x$ for $x\in [0,1]$. Hence
\begin{align*}
    J(x,n_2) \leq J(1,n_2) = \frac{(n_1-1)!}{(n_1-n_2)!(n_2-1)!} (n_2-1)\int_0^1t^{n_1-n_2}(1-t)^{n_2-2}dt  =  1 \nonumber.
\end{align*}
(ii). The key idea is to evaluate $J(x,n_2)-J(x,n_2-1)$. Observe that for $3\leq n_2\leq n_1$, 
\begin{align}
     J(x,n_2-1) &= \binom{n_1-1}{n_2-2} \biggl ((n_2-2)\int_0^{x}t^{n_1-n_2+1}(1-t)^{n_2-3}dt + x^{n_1-n_2+1}(1-x)^{n_2-2} \biggl) \nonumber \\
     &=\binom{n_1-1}{n_2-2} (n_1-n_2+1)\int_0^{x}t^{n_1-n_2}(1-t)^{n_2-2}dt \nonumber \\
     &=\binom{n_1-1}{n_2-1} (n_2-1)\int_0^{x}t^{n_1-n_2}(1-t)^{n_2-2}dt \nonumber 
\end{align}
then we have $\forall x \in [0,1]$,
\begin{align}
     J(x,n_2-1)- J(x,n_2)&=-\binom{n_1-1}{n_2-1}x^{n_1-n_2}(1-x)^{n_2-1}<0\nonumber. 
\end{align}
Hence $J(x,n_2)$ is increasing for $n_2\in \left[2,n_1\right]$.
\end{proof}

\begin{proof}[\textbf{Proof of Proposition \ref{prop:monotonicity of v}}]
For brevity, we write $J(x, n_2)$ as $J(x)$ in the following proof. 

\medskip

\noindent\textit{\underline{Part 1: Proofs of (i) and the equivalence of (a) and (c) in (ii)}}.

\medskip

\noindent We begin by calculating $v'(x\mid n_2)$ as follows.
\begin{align*}
    v'(x\mid n_2) &=  \left(\frac{x}{J(F(x))}\right)'  = \frac{J(F(x)) - xf(x)J'(F(x))}{J^2(F(x))}.
\end{align*}
Let $W(x) = J(F(x)) - xf(x)J'(F(x))$, then $W'(x) = -x\left(f'(x)J'(F(x)) + f^2(x)J''(F(x))\right)$. Observe that 
\begin{align*}
    J'(x) &= \binom{n_1-1}{n_2-1} (n_1-n_2)x^{n_1-n_2-1}\bigl(1 - x \bigl)^{n_2-1},\\
    J''(x) &= \binom{n_1-1}{n_2-1} (n_1-n_2)x^{n_1-n_2-2}\bigl(1 - x \bigl)^{n_2-2}(n_1 - n_2 - 1 - (n_1 - 2)x),
\end{align*}
hence when $F(x) = x^\theta$, $$W'(x) = \binom{n_1-1}{n_2-1} (n_1-n_2) \theta x^{(n_1-n_2)\theta-1}\left(1 - x^\theta\right)^{n_2 - 2} \left(((n_1 + 1)\theta - 1)x^\theta - ((n_1 - n_2)\theta - 1) \right).$$

\noindent(1). When $(n_1 - n_2)\theta \leq 1$, then $\left(((n_1 - 1)\theta - 1)x^\theta - ((n_1 - n_2)\theta - 1) \right) \mid _{x=0} = 1 - (n_1 - n_2)\theta \geq 0$, and $\left(((n_1 - 1)\theta - 1)x^\theta - ((n_1 - n_2)\theta - 1) \right) \mid _{x=1} = (n_2 - 1)\theta > 0$. Observe that $((n_1 - 1)\theta - 1)x^\theta - ((n_1 - n_2)\theta - 1)$ is monotone in $[0, 1]$, we conclude that $W'(x)\geq 0$ in $[0, 1]$. Since $W(0) = 0$, we have $W(x) \geq 0$, and $v(a_i\mid n_2)$ is increasing in $[0, 1]$.

\medskip
\noindent(2). When $(n_1 - n_2)\theta > 1$, then $W(x)$ is decreasing in $\left[0, \left(\frac{(n_1 - n_2)\theta - 1)}{(n_1 - 1)\theta - 1}\right)^{\frac{1}{\theta}}\right]$, and increasing in $\left[ \left(\frac{(n_1 - n_2)\theta - 1)}{(n_1 - 1)\theta - 1}\right)^{\frac{1}{\theta}}, 1\right]$. Since $W(0) = 0$, and $W(1) = 1$, there exists $\hat{a} \in [0, 1)$ such that $W(x) \leq 0$ in $[0, \hat{a}]$ and $W(x) \geq 0$ in $[\hat{a}, 1]$, which implies that $v(a_i\mid n_2)$ decreases in $[0, \hat{a}]$ and increases in $[\hat{a}, 1]$.

\medskip
\noindent Steps (1) and (2) complete the proof of (i) and the equivalence of (a) and (c) in (ii).

\medskip

\noindent\textit{\underline{Part 2: Equivalence of (b) and (c) in (ii)}}.

\medskip

\noindent By L'Hôpital's  rule, we have
\begin{align*}
    v_i(0 \mid n_2)  = \lim_{a_i\rightarrow 0}\frac{a_i}{J(F(a_i))} = \frac{1}{\binom{n_1+1}{n_2-1} (n_1-n_2)}\lim_{a_i\rightarrow 0} \frac{1}{f(a_i)\cdot F^{n_1-n_2-1}(a_i)\cdot (1-F(a_i))^{n_2-1}},
\end{align*}
hence when $F(x) = x^\theta$, 
\begin{align*}
    v_i(0 \mid n_2) \text{ is finite}~\Longleftrightarrow~\lim_{a_i\rightarrow 0} \frac{1}{x^{(n_1-n_2)\theta - 1}} \text{ exists}~\Longleftrightarrow~(n_1 - n_2)\theta \leq 1.
\end{align*}
This completes the proof.
\end{proof}

\medskip

\begin{proof}[\textbf{Proof of Theorem \ref{thm_equilibrium strategy_general_prize}}]
The proof includes four steps: i) derive the unique symmetric and strictly increasing Bayesian Nash equilibrium $b(a_i\mid n_2)$ if exists; ii) check the obtained function $b(a_i\mid n_2)$ is indeed increasing; iii) show that $b(a_i\mid n_2)$ is indeed an equilibrium under the condition that $v(\cdot\mid n_2)$ is increasing; iv) if there exists strictly increasing Bayesian Nash equilibrium, $v(\cdot\mid n_2)$ must be increasing.

\smallskip

\noindent\textit{\underline{Step 1: Symmetric and strictly increasing BNE if exists}}

\smallskip

\noindent When all $n_2-1$ other admitted players follow a strictly increasing equilibrium strategy $b(\cdot)$ and player $i$ exerts effort $e_i = b(\tilde{a}_i)$ for some $\tilde{a}_i\in[0,1]$, the objective of player $i \in \mathcal{I}$ is to
\begin{equation}
\label{Eq_Equivalent Objective Function original}
        \max_{e_i}~u_{i} := \sum_{\ell=1}^{n_2} V_\ell P_{i,\ell}(e_i\mid b) - \frac{g(e_i)}{a_i} ,
\end{equation} 
where $P_{i,\ell}(e_i\mid b)$ is the probability that player $i$ wins the $\ell$-th prize, and 
\begin{align}
    P_{i,\ell}(e_i\mid b) & = \Pr\left(\text{$e_i$ ranks $\ell^{\text{th}}$ highest in $\mathcal{I}$} \right) \nonumber = \Pr\left(e_{i} < b(A_s\mid n_2), s \in \Tilde{\mathcal{I}}; e_{i}> b(A_k\mid n_2), k \in \mathcal{I}\setminus \Tilde{\mathcal{I}}\right)   \nonumber\\
    & = \Pr\left(\tilde{a}_i < A_s, s \in \Tilde{\mathcal{I}}; \tilde{a}_i > A_k, k \in \mathcal{I}\setminus \Tilde{\mathcal{I}}\right)  \nonumber = \tbinom{n_2-1}{\ell-1} \underbrace{\int_{\tilde{a}_i}^1\cdots \int_{\tilde{a}_i}^1}_{\ell-1} \underbrace{\int_0^{\tilde{a}_i}\cdots\int_0^{\tilde{a}_i}}_{n_2-\ell} \beta_{i}\left(a_{-i}\mid s,a_{i}\right) da_{-i} \nonumber. 
\end{align} 
Since $P_{i,\ell}(e_i\mid b)$ only depends on $\tilde{a}_i$ and $a_i$, we rewrite $P_{i,\ell}(e_i\mid b)$ as $P_{i,\ell}(\tilde{a}_i,a_i)$. Equation (\ref{Eq_Equivalent Objective Function original}) is equivalent to 
\begin{equation}
\label{Eq_Equivalent Objective Function}
\max_{\tilde{a}_i}~U_i(\tilde{a}_i,a_i\mid n_2, b):= a_i\sum_{\ell=1}^{n_2} V_\ell P_{i,\ell}(\tilde{a}_i,a_i) - g(b(\tilde{a}_i\mid n_2)) ,
\end{equation} 
We now calculates $P_{i,\ell}(\tilde{a}_i,a_i)$ as follows.

\vspace{0.5em}

\noindent\textit{\underline{Case 1: $\ell<n_2$}}.
\medskip

\noindent(i). When $\tilde{a}_i\leq a_i$,
\begin{align}
    P_{i,\ell}(\tilde{a}_i,a_i) & = \tbinom{n_2-1}{\ell-1} \underbrace{\int_{\tilde{a}_i}^1\cdots \int_{\tilde{a}_i}^1}_{\ell-1} \underbrace{\int_0^{\tilde{a}_i}\cdots\int_0^{\tilde{a}_i}}_{n_2-\ell} \beta_{i}\left(a_{-i}\mid s,a_{i}\right) da_{-i} \nonumber\\
        & = \tbinom{n_2-1}{\ell-1} \left (\prod_{s\in \Tilde{\mathcal{I}}}\int_{\tilde{a}_i}^1 f(a_s)da_s\right) \underbrace{\int_0^{\tilde{a}_i}\cdots\int_0^{\tilde{a}_i}}_{n_2-\ell}   \frac{F^{n_1-n_2}\left(a_{(1)}\right)}{I(F(a_i),n_2)}\prod_{k\in \mathcal{I}_{-i}\setminus \Tilde{\mathcal{I}}}\left(f(a_k)da_k \right)\nonumber  \\
        & = \tbinom{n_2-1}{\ell-1} \frac{\bigl(1 - F(\tilde{a}_i) \bigl)^{\ell-1}}{I(F(a_i),n_2)} \underbrace{F^{n_1-n_2+n_2-\ell}(\tilde{a}_i)  \frac{(n_1-n_2)!(n_2-\ell)!}{(n_1-n_2+n_2-\ell)!}}_{\text{Lemma \ref{Auxiliary Lemma_Integral}}} \label{Eq_Auxiliary Lemma_Seq-Equlibrium proof} \\
    & = \frac{1}{I(F(a_i),n_2)} \tbinom{n_2-1}{\ell-1}\tbinom{n_1-\ell}{n_2-\ell}\big(1 - F(\tilde{a}_i) \big)^{\ell-1} F^{n_1-\ell}(\tilde{a}_i) \nonumber \\
    &= \frac{1}{J(F(a_i),n_2)} \left({F}_{(\ell,n_1-1)}(\tilde{a}_i) - {F}_{(\ell - 1,n_1-1)}(\tilde{a}_i)\right)\label{eq:l<n_2 left}
\end{align} 
where $a_{(1)}  = \min \left\{a_j: j\in \mathcal{I}_{-i} \right\}$ and $J(F(a_i),n_2)  := \binom{n_1-1}{n_2-1}\cdot I(F(a_i),n_2)$.
Since $\max \left\{a_j:j\in \mathcal{I}_{-i}\setminus \Tilde{\mathcal{I}} \right\} < \tilde{a}_i$ and $\min \left\{a_j: j\in \Tilde{\mathcal{I}}\right\} > \tilde{a}_i$, we have $a_{(1)} = \min \left\{a_j:j\in \mathcal{I}_{-i}\setminus \Tilde{\mathcal{I}} \right\}$, then Lemma \ref{Auxiliary Lemma_Integral} can be applied in equation (\ref{Eq_Auxiliary Lemma_Seq-Equlibrium proof}).

\medskip

\noindent(ii). When $\tilde{a}_i\geq a_i$, by the similar calculation of step (i) and Lemma \ref{Auxiliary Lemma_Integral}, we have 
\begin{align}
    P_{i,\ell}(\tilde{a}_i,a_i)= & \tbinom{n_2-1}{\ell-1}\frac{\bigl(1 - F(\tilde{a}_i) \bigl)^{\ell-1} }{I(F(a_i),n_2)}Q(a_i,\tilde{a}_i\mid n_1,n_2,n_2 - \ell)\label{eq:l<n_2 right}\\
    =&\tbinom{n_2-1}{\ell-1}\frac{\bigl(1 - F(\tilde{a}_i) \bigl)^{\ell-1} }{I(F(a_i),n_2)} \int_{0}^{F(a_i)}(F(\tilde{a}_i) - t)^{n_2-\ell}dt^{n_1-n_2}\nonumber.
\end{align} 

\vspace{0.5em}
\noindent\textit{\underline{Case 2: $\ell=n_2$}}. 

\medskip
\noindent(iii). When $\tilde{a}_i\leq a_i$, Lemma \ref{Auxiliary Lemma_Integral} gives that  
\begin{align}
    P_{i,\ell}(\tilde{a}_i,a_i) & = \tbinom{n_2-1}{n_2-1} \underbrace{\int_{\tilde{a}_i}^1\cdots \int_{\tilde{a}_i}^1}_{n_2-1}  \beta_{i}\left(a_{-i}\mid s,a_{i}\right) da_{-i}=\frac{R(\tilde{a}_i, a_i\mid n_1,n_2,n_2 - 1)}{I(F(a_i),n_2)} \label{eq:l=n_2 left}.
\end{align} 

\medskip
\noindent(iv).
When $\tilde{a}_i\geq a_i$,
\begin{align}
    P_{i,\ell}(\tilde{a}_i,a_i) & = \tbinom{n_2-1}{n_2-1} \underbrace{\int_{\tilde{a}_i}^1\cdots \int_{\tilde{a}_i}^1}_{n_2-1}  \beta_{i}\left(a_{-i}\mid s,a_{i}\right) da_{-i}=\underbrace{\int_{\tilde{a}_i}^1\cdots \int_{\tilde{a}_i}^1}_{n_2-1}  \frac{F^{n_1-n_2}(a_i)}{I(F(a_i),n_2)}\prod_{k \in \mathcal{I}_{-i}}\left(f(a_k)da_k\right) \nonumber  \\
    & = \frac{F^{n_1-n_2}(a_i)}{I(F(a_i),n_2)}\left((1 - F(\tilde{a}_i) \right)^{n_2-1} \label{eq:l=n_2 right}.
\end{align} 
If there exists a unique strictly increasing Bayesian Nash equilibrium strategy, by the first order condition, $b(a_i\mid n_2)$ must satisfy that 
\begin{align*}
    \frac{\partial}{\partial \tilde{a}_i} U_i(\tilde{a}_i,a_i\mid n_2, b) \bigg|_{\tilde{a}_i = a_i} = \left(a_i\sum_{\ell = 1}^{n_2}V_{\ell}\frac{\partial P_{i,\ell}(\tilde{a}_i,a_i)}{\partial \tilde{a}_i} - \frac{\partial g(b(\tilde{a}_i\mid n_2))}{\partial \tilde{a}_i} \right) \bigg |_{\tilde{a}_i = a_i} = 0.
\end{align*}
Since $\beta_{i}\left(a_{-i}\mid s,a_{i}\right)$ is continuous, $U_{i}(\tilde{a}_i,a_i)$ is differentiable with respect to $\tilde{a}_i$, calculating $\frac{\partial}{\partial \tilde{a}_i} U_i(\tilde{a}_i,a_i\mid n_2, b)$ is equivalent to calculating the left or right derivative, which gives 
\begin{align*}
    &\frac{\partial}{\partial \tilde{a}_i}U_i(\tilde{a}_i,a_i\mid n_2, b)\\
    =& v(a_i\mid n_2) \sum_{\ell = 1}^{n_2 - 1} V_{\ell}\left(F'_{(\ell,n_1-1)}(a_i) - F'_{(\ell-1,n_1-1)}(a_i)\right) + \frac{a_iF^{n_1-n_2}(a_i)}{I(F(a_i),n_2)}V_{n_2}\left(\left((1 - F(\tilde{a}_i) \right)^{n_2-1}\right)' - \left(g(b(a_i\mid n_2))\right)'.
\end{align*}
Hence $\frac{\partial}{\partial \tilde{a}_i}U_i(\tilde{a}_i,a_i\mid n_2, b)\bigg|_{\tilde{a}_i = a_i} = 0$ implies that 
\begin{align*}
    \left(g(b(a_i\mid n_2))\right)' 
    & = v(a_i\mid n_2)  \sum_{\ell = 1}^{n_2 - 1} V_{\ell}\left(F'_{(\ell,n_1-1)}(a_i) - F'_{(\ell-1,n_1-1)}(a_i)\right) - v(a_i\mid n_2) F'_{(n_2 - 1,n_1-1)}(a_i)\\
    & = v(a_i\mid n_2)\sum\limits_{\ell=1}^{n_2-1}(V_\ell-V_{\ell+1})F'_{(\ell,n_1-1)}(a_i),
\end{align*}
where the penultimate equality uses the fact that 
\begin{align*}
    \frac{a_i F^{n_1-n_2}(a_i)}{I(F(a_i),n_2)} \left(\left(1 - F(a_i)\right)^{n_2-1}\right)'
    &= (n_2 - 1) v(a_i\mid n_2) \tbinom{n_1-1}{n_2-1}V_{n_2} F^{n_1-n_2}(a_i)\left((1 - F(a_i) \right)^{n_2-2}f(a_i)\\
    &= v(a_i\mid n_2) F'_{(n_2 - 1,n_1-1)}(a_i).
\end{align*}
Since $b(0\mid n_2) = 0$, the equilibrium effort can be given as follows.
\begin{align*}
    b(a_i\mid n_2) =g^{-1}\left(   \int_{0}^{a_i}v(x\mid n_2)\sum\limits_{\ell=1}^{n_2-1}(V_\ell-V_{\ell+1})dF_{(\ell,n_1-1)}(x)\right).
\end{align*}

\noindent\textit{\underline{Step 2: $b(a_i\mid n_2)$ is indeed strictly increasing.}}

\medskip

\noindent This can be immediately checked by observing that when the prizes are not all identical,
\begin{align*}
    \left(\int_{0}^{a_i}v(x\mid n_2)\sum\limits_{\ell=1}^{n_2-1}(V_\ell-V_{\ell+1})dF_{(\ell,n_1-1)}(x)\right)' = v(a_i\mid n_2)\sum\limits_{\ell=1}^{n_2-1}(V_\ell-V_{\ell+1})dF_{(\ell,n_1-1)}(a_i)>0.
\end{align*}
\noindent\textit{\underline{Step 3: $b(a_i\mid n_2)$ is indeed an equilibrium under the condition $v(a_i \mid n_2)$ is increasing.}}

\medskip

\noindent To show this, it's equivalent to show that
\begin{align}\label{eq: utility_equilibrium_check}
    U_i(\tilde{a}_i,a_i\mid n_2, b) = a_i\sum_{\ell=1}^{n_2} V_\ell P_{i,\ell}(\tilde{a}_i, a_i) - g(b(\tilde{a}_i\mid n_2))
\end{align}
indeed obtains the maximum at $\tilde{a}_i = a_i$. We will show that $U_i(\tilde{a}_i, a_i)$ increases with $\tilde{a}_i \in [0, a_i]$, and decreases with $\tilde{a}_i \in [a_i, 1]$.

\noindent(a). When $\tilde{a}_i\leq a_i$, plug equations (\ref{eq:l<n_2 left}) and (\ref{eq:l=n_2 left}) into equation (\ref{eq: utility_equilibrium_check}), we obtain
\begin{align*}
    U_i(\tilde{a}_i,a_i\mid n_2, b) = &a_i\sum_{\ell=1}^{n_2-1}  \frac{V_\ell}{J(F(a_i),n_2)} \left({F}_{(\ell,n_1-1)}(\tilde{a}_i) - {F}_{(\ell + 1,n_1-1)}(\tilde{a}_i)\right) \\
    &+ \frac{a_iV_{n_2}}{I(F(a_i),n_2)}R(\tilde{a}_i, a_i\mid n_1,n_2,n_2 - 1) 
    - \int_{0}^{a_i}v(x\mid n_2)\sum\limits_{\ell=1}^{n_2-1}(V_\ell-V_{\ell+1})dF_{(\ell,n_1-1)}(x),
\end{align*}
where the definition of $R(\cdot, \cdot\mid \cdot,\cdot,\cdot)$ is given by Lemma \ref{Auxiliary Lemma_Integral}.
Observe that by Lemma \ref{Auxiliary Lemma_Integral} and the monotonically increasing property of $v(\cdot \mid n_2)$,
\begin{align}
    &\frac{\partial U_i(\tilde{a}_i,a_i\mid n_2, b)}{\partial \tilde{a}_i}\nonumber\\
    =& v(a_i\mid n_2)\sum_{\ell=1}^{n_2-1} V_{\ell} \bigl(F'_{(\ell,n_1-1)}(\tilde{a}_i) - F'_{(\ell-1,n_1-1)}(\tilde{a}_i)\bigl) -  \frac{ (n_2-1) a_i V_{n_2}}{I(F(x),n_2)}F^{n_1-n_2}(\tilde{a}_i) \bigl(1 - F(\tilde{a}_i)\bigl)^{n_2-2} f(\tilde{a}_i) \biggl)\nonumber\\
    &- v(\tilde{a}_i\mid n_2)\sum\limits_{\ell=1}^{n_2-1}(V_\ell-V_{\ell+1})F'_{(\ell,n_1-1)}(\tilde{a}_i)\nonumber\\
    =& \left(v({a}_i\mid n_2) - v(\tilde{a}_i\mid n_2) \right)\sum\limits_{\ell=1}^{n_2-1}(V_\ell-V_{\ell+1})F'_{(\ell,n_1-1)}(\tilde{a}_i)\geq 0 \label{eq:partial derivative of U}.
\end{align}
Hence $U_i(\tilde{a}_i, a_i)$ is increasing in $\tilde{a}_i \in [0, a_i]$.

\medskip
\noindent(b). When $\tilde{a}_i\geq a_i$, plug equations (\ref{eq:l<n_2 right}) and (\ref{eq:l=n_2 right}) into equation (\ref{eq: utility_equilibrium_check}), we obtain
\begin{align*}
    U_i(\tilde{a}_i,a_i\mid n_2, b) = &a_i\sum_{\ell=1}^{n_2-1}  \tbinom{n_2-1}{\ell-1}\frac{V_{\ell}}{I(F(a_i),n_2)}\bigl(1 - F(\tilde{a}_i) \bigl)^{\ell-1} Q(a_i,\tilde{a}_i\mid n_1,n_2,n_2 - \ell) \\
    &+ \frac{a_iV_{n_2}}{I(F(a_i),n_2)}F^{n_1-n_2}(a_i)\left((1 - F(\tilde{a}_i) \right)^{n_2-1} - \int_{0}^{a_i}v(x\mid n_2)\sum\limits_{\ell=1}^{n_2-1}(V_\ell-V_{\ell+1})F'_{(\ell,n_1-1)}(x)\\
    =&v(a_i\mid n_2)\sum_{\ell=1}^{n_2-1}  \tbinom{n_2-1}{\ell-1}\tbinom{n_1-1}{n_2-1}V_{\ell}\bigl(1 - F(\tilde{a}_i) \bigl)^{\ell-1} Q(a_i,\tilde{a}_i\mid n_1,n_2,n_2 - \ell) \\
    &+  v(a_i\mid n_2) \tbinom{n_1-1}{n_2-1}V_{n_2} F^{n_1-n_2}(a_i)\left((1 - F(\tilde{a}_i) \right)^{n_2-1} - \int_{0}^{a_i}v(x\mid n_2)\sum\limits_{\ell=1}^{n_2-1}(V_\ell-V_{\ell+1})F'_{(\ell,n_1-1)}(x)\\
     = &v(a_i\mid n_2)S(a_i,\tilde{a}_i\mid n_1,n_2) - \int_{0}^{a_i}v(x\mid n_2)\sum\limits_{\ell=1}^{n_2-1}(V_\ell-V_{\ell+1})F'_{(\ell,n_1-1)}(x),
\end{align*}
where
\begin{align*}
    &S(a_i,\tilde{a}_i\mid n_1,n_2)\\ = &\sum_{\ell=1}^{n_2-1}  \tbinom{n_2-1}{\ell-1}\tbinom{n_1-1}{n_2-1}V_{\ell}\bigl(1 - F(\tilde{a}_i) \bigl)^{\ell-1} Q(a_i,\tilde{a}_i\mid n_1,n_2,n_2 - \ell)+ \tbinom{n_1-1}{n_2-1}V_{n_2} F^{n_1-n_2}(a_i)\left((1 - F(\tilde{a}_i) \right)^{n_2-1},
\end{align*}
and the definition of $Q(\cdot, \cdot\mid \cdot,\cdot,\cdot)$ is given by Lemma \ref{Auxiliary Lemma_Integral}.
By Lemma \ref{Auxiliary Lemma_prize}, we can conclude that 
\begin{align*}
    \frac{\partial S(a_i,\tilde{a}_i\mid n_1,n_2)}{\partial \tilde{a}_i} \leq & \frac{\partial S(\tilde{a}_i,\tilde{a}_i\mid n_1,n_2)}{\partial \tilde{a}_i} = \sum\limits_{\ell=1}^{n_2-1}(V_\ell-V_{\ell+1})F'_{(\ell,n_1-1)}(\tilde{a}_i).
\end{align*}
Therefore, 
\begin{align*}
    \frac{\partial U_i(\tilde{a}_i,a_i\mid n_2, b)}{\partial \tilde{a}_i}
    &\leq v(a_i\mid n_2)\sum\limits_{\ell=1}^{n_2-1}(V_\ell-V_{\ell+1})F'_{(\ell,n_1-1)}(\tilde{a}_i) - v(\tilde{a}_i\mid n_2)\sum\limits_{\ell=1}^{n_2-1}(V_\ell-V_{\ell+1})F'_{(\ell,n_1-1)}(\tilde{a}_i) \\
    &= (v(a_i\mid n_2) - v(\tilde{a}_i\mid n_2)) \sum\limits_{\ell=1}^{n_2-1}(V_\ell-V_{\ell+1})F'_{(\ell,n_1-1)}(\tilde{a}_i) \leq 0,
\end{align*}
where the last inequality uses the monotonically increasing property of $v(\cdot \mid n_2)$. This indicates that $U_i(\tilde{a}_i, a_i)$ is decreasing in $\tilde{a}_i \in [a_i, 1]$.
Combining (a) and (b), we can conclude that $U_i(\tilde{a}_i, a_i)$ indeed obtains the maximum at $\tilde{a}_i = a_i$.

\noindent\textit{\underline{Step 4: If there exists strictly increasing Bayesian Nash equilibrium, $v(\cdot\mid n_2)$ must be increasing.}}

\medskip

\noindent From step 1, we know that if there exists strictly increasing Bayesian Nash equilibrium, the equilibrium must be $b(a_i\mid n_2)$ given in equation (\ref{eq_b(a_i)_general_prize}). Hence we need to show that $b(a_i\mid n_2)$ is the equilibrium implying $v(\cdot\mid n_2)$ is increasing. We derive this by making a contradiction.

\noindent If $v(\cdot\mid n_2)$ is not increasing, then there exists some interval $[\underline{a}, \overline{a}]$ on which $v(\cdot\mid n_2)$ is strictly decreasing. Pick $a_i \in (\underline{a}, \overline{a})$, it's sufficient to show that 
\begin{align*}
    \text{argmax}_{\tilde{a}_i}\frac{\partial U_i(\tilde{a}_i,a_i\mid n_2, b)}{\partial \tilde{a}_i} \neq a_i.
\end{align*}
Indeed, when $\tilde{a}_i \leq a_i$, by equation (\ref{eq:partial derivative of U}), the fact $v(\cdot\mid n_2)$ is strictly decreasing on $[\underline{a}, \overline{a}]$, and Lagrangian's mean value theorem, we can obtain 
\begin{align*}
    &U_i(\underline{a}, a_i\mid n_2, b) - U_i(a_i, a_i\mid n_2, b)\\
    =& \frac{\partial U_i(\tilde{a}_i, a_i\mid n_2, b)}{\partial \tilde{a}_i}\bigg|_{\tilde{a}_i = \xi} (\underline{a} - a_i) = (\underline{a} - a_i)\left(v({a}_i\mid n_2) - v(\xi\mid n_2) \right)\sum\limits_{\ell=1}^{n_2-1}(V_\ell-V_{\ell+1})F'_{(\ell,n_1-1)}(\xi) > 0,
\end{align*}
where $\xi \in (\underline{a}, a_i)$.

\medskip

\noindent Steps 1 to 4 complete the proof.
\end{proof}

\medskip

\begin{proof}[\textbf{Proof of Corollary \ref{coro:comparison_n_2_n_2_1}}]
By Lemma \ref{lemma_J_increasing_n_2}, we know that given $n_1$, for any $x\in [0,1]$, $J(x,n_2)$ is increasing in $n_2\in[2,n_1]$, hence $v(x\mid n_2)$ is decreasing in $n_2\in[2,n_1]$.
When $V_{n_2} = 0$, combining the formula
\begin{equation*}
b(a_i\mid n_2) = 
   g^{-1}\left(   \int_{0}^{a_i}v(x\mid n_2)\sum\limits_{\ell=1}^{n_2-1}(V_\ell-V_{\ell+1})dF_{(\ell,n_1-1)}(x)\right),
\end{equation*} 
$g(\cdot)$ is increasing, and  $dF_{(\ell,n_1-1)}(x)>0 $ for $1\leq \ell \leq n_2 - 1$, we conclude that $b(a_i\mid n_2 - 1)\geq b(a_i\mid n_2)$, which completes the proof.
\end{proof}

\medskip

\begin{proof}[\textbf{Proof of Theorem \ref{thm:opt_number}}]
Under the conditions (i) and (ii), by Theorem \ref{thm_equilibrium strategy_general_prize}, 
\begin{equation*}
b(a_i\mid n_2) = 
   g^{-1}\left(   \int_{0}^{a_i}v(x\mid n_2)(V_1-V_{2})dF_{(\ell,n_1-1)}(x)\right).
\end{equation*} 
Similar to the discussion in Corollary \ref{coro:comparison_n_2_n_2_1}, we have $v(x\mid n_2)$ is decreasing in $n_2\in[2,n_1]$, hence $b(a_i\mid n_2)$ obtains the maximum when $n_2 = 2$, $V_1 = 1$, and $V_2 = 0$. Theorem \ref{thm:opt_number} then follows by observing that the expected highest effort is $\int_0^1 b(a_i\mid n_2)\,dF^{n_1}(a_i)$.
\end{proof}

\medskip

\section{Auxiliary Results}
\label{sec_appendix_auxiliaryreults}

We provide 
a useful identity (Appendix \ref{sec_appendix_identity}) in this section.

\label{sec_appendix_identity}

\label{sec_appendix_identity}
\begin{lemma}
\label{Auxiliary Lemma_Integral}
For any distribution function $F(\cdot)$ with continuous density $f(\cdot)$ in the support of $[0,1]$, $m\leq n$, and $0\leq x\leq y \leq 1$,
\begin{align*}
   Q(x,y\mid n,m,k) = &\underbrace{\int_0^y\cdots \int_0^y}_{k} \left(F^{n-m}(t_{(1)})\wedge F^{n - m}(x)\right) \prod_{i=1}^kf(t_i) dt_1\cdots dt_k\\
   =& \sum_{j=0}^{k}\frac{\binom{k}{j}}{\binom{n-m+j}{j}}F^{n-m+j}(x)(F(y) - F(x))^{k-j}\\
   =& \int_{0}^{F(x)}(F(y) - t)^kdt^{n-m}, 
\end{align*}
where $t_{(1)} = \normalfont{\min}\{t_1,t_2,\cdots,t_k\}$.
Specifically, 
\[
Q(x,x\mid n,m,k) = \underbrace{\int_0^x\cdots \int_0^x}_{k} F^{n-m}(t_{(1)})\prod_{i=1}^kf(t_i) dt_1\cdots dt_k =\frac{1}{\binom{n-m+k}{k}}F^{n-m+k}(x) ,\ \  \forall x\in [0,1].
\]
In addition, 
\begin{align*}
   R(x,y\mid n,m,k) = &\underbrace{\int_x^1\cdots \int_x^1}_{k} \left(F^{n-m}(t_{(1)})\wedge F^{n - m}(y)\right) \prod_{i=1}^kf(t_i) dt_1\cdots dt_k\\
   = &\sum_{j=1}^{k}\binom{k}{j}j(F(y) - F(x))^{k-j}\int_{F(x)}^{F(y)}t^{n-m}\left(F(y)-t \right)^{j-1}dt + F^{n-m}(y)(1 - F(y))^k.
\end{align*}
\end{lemma}
 
\begin{proof}
We have
\begin{align*}
     &Q(x,y\mid n,m,k)\nonumber\\ = &\sum_{j=0}^{k} \binom{k}{j}\underbrace{\int_0^x\cdots \int_0^x}_{j}\underbrace{\int_{x}^y\cdots \int_{x}^y}_{k-j} \left(F^{n-m}(t_{(1)})\wedge F^{n-m}(x) \right)\prod_{i=1}^k(f(t_i) dt_i)
     \\=&\sum_{j=1}^{k}\binom{k}{j}(F(y) - F(x))^{k-j}\frac{j!}{(j-1)!}\int_0^x\bigg(\underbrace{\int_{t_1}^x\cdots \int_{t_1}^x}_{j-1} F^{n-m}(t_{1}) \prod_{i=1}^jf(t_i) dt_2\cdots dt_j \bigg)dt_1 + F^{n-m}(x)(F(y) - F(x))^k
     \\=&\sum_{j=1}^{k}\binom{k}{j}j(F(y) - F(x))^{k-j}\int_0^xF^{n-m}(t_{1})\left(F(x)-F(t_1) \right)^{j-1}dF(t_1) + F^{n-m}(x)(F(y) - F(x))^k\\
     =& \sum_{j=1}^{k}\binom{k}{j}j(F(y) - F(x))^{k-j}F^{n-m+j}(x)\int_0^1t^{n-m}\left(1-t\right)^{j-1}dt + F^{n-m}(x)(F(y) - F(x))^k
     \\=& \sum_{j=1}^{k}\frac{\binom{k}{j}}{\binom{n-m+j}{j}}F^{n-m+j}(x)(F(y) - F(x))^{k-j} + F^{n-m}(x)(F(y) - F(x))^k 
     \\=& \sum_{j=0}^{k}\frac{\binom{k}{j}}{\binom{n-m+j}{j}}F^{n-m+j}(x)(F(y) - F(x))^{k-j}\\
     =& \int_{0}^{F(x)}(F(y) - t)^kdt^{n-m}.
\end{align*}
Specifically, when $y = x$,
\begin{align*}
    Q(x,x\mid n,m,k) 
     = \frac{1}{\binom{n-m+k}{k}}F^{n-m+k}(x).
\end{align*}
In addition, 
\begin{align*}
     &R(x,y\mid n,m,k)\nonumber\\ = &\sum_{j=0}^{k} \binom{k}{j}\underbrace{\int_x^y\cdots \int_x^y}_{j}\underbrace{\int_{y}^1\cdots \int_{y}^1}_{k-j} \left(F^{n-m}(t_{(1)})\wedge F^{n-m}(y) \right)\prod_{i=1}^k(f(t_i) dt_i)
     \\=&\sum_{j=1}^{k}\binom{k}{j}(1 - F(y))^{k-j}\frac{j!}{(j-1)!}\int_x^y\bigg(\underbrace{\int_{t_1}^y\cdots \int_{t_1}^y}_{j-1} F^{n-m}(t_{1}) \prod_{i=1}^jf(t_i) dt_2\cdots dt_j \bigg)dt_1 + F^{n-m}(y)(1 - F(y))^k
     \\=&\sum_{j=1}^{k}\binom{k}{j}j(F(y) - F(x))^{k-j}\int_x^yF^{n-m}(t_{1})\left(F(y)-F(t_1) \right)^{j-1}dF(t_1) + F^{n-m}(y)(1 - F(y))^k\\
     =& \sum_{j=1}^{k}\binom{k}{j}j(F(y) - F(x))^{k-j}\int_{F(x)}^{F(y)}t^{n-m}\left(F(y)-t \right)^{j-1}dt + F^{n-m}(y)(1 - F(y))^k.
\end{align*}
This completes the proof.
\end{proof}

\medskip

\begin{lemma}
\label{Auxiliary Lemma_derivative}
Let $F(\cdot)$, $Q(\cdot,\cdot)$, $n$, $m$ and 
$k$ be the same as defined in Lemma \ref{Auxiliary Lemma_Integral}.
For $0 \leq x \leq y \leq 1 $, $\ell \leq m$, 
\begin{itemize}
    \item $Q(x,y)\geq Q(x,x) = \frac{1}{\binom{n-m+k}{k}}F^{n-m+k}(x)$, 
    \item $\frac{\partial Q(x,y\mid n,m,k)}{\partial y}\leq \frac{1}{\binom{n-m+k}{k}}\left(F^{n - m + k}(y)\right)'$, where the equality holds when $x = y$.
\end{itemize}
 
\end{lemma}
\begin{proof}
    Observe that 
    \begin{align*}
        \frac{\partial Q}{\partial y}(x,y\mid n,m,k) = kf(y)  \int_{0}^{F(x)}(F(y) - t)^{k-1}dt^{n-m}\geq0,
    \end{align*}
hence $Q(x,y\mid n,m,k)$ is increasing in $y \in [x,1]$, which implies that $Q(x,y)\geq Q(x,x)$.

\noindent Furthermore, we have 
\begin{align*}
    \frac{\partial^2 Q(x,y\mid n,m,k)}{\partial y \partial x} = k(n-m) F^{n-m-1}(x)(F(y) - F(x))^{k-1}f(x)f(y) \geq 0.
\end{align*}
Then $\frac{\partial Q(x,y\mid n,m,k)}{\partial y}$ is increasing with $x$ in $[0,y]$, which indicates that 
\begin{align*}
    \frac{\partial Q(x,y\mid n,m,k)}{\partial y} &\leq \frac{\partial Q(y,y\mid n,m,k)}{\partial y} = k(n-m)f(y)\int_0^{F(y)}t^{n-m-1}\left(F(y)-t \right)^{k-1}dt\\
    &=k(n-m)f(y)F^{n-m+k-1}(y)\int_0^1t^{n-m-1}\left(1-t \right)^{k-1}dt\\
    & = \frac{1}{\binom{n-m+k}{k}}\left(F^{n - m + k}(y)\right)'.
\end{align*}
This completes the proof of the lemma.
\end{proof}

\begin{lemma}
\label{Auxiliary Lemma_prize}
Let $n_2 \leq n_1$. For $F(\cdot)$ and $Q(\cdot,\cdot)$ defined in Lemma \ref{Auxiliary Lemma_Integral}, and $0 \leq a_i \leq \tilde{a}_i \leq 1 $,
define 
\begin{align*}
&S(a_i,\tilde{a}_i\mid n_1,n_2)\\ = &\sum_{\ell=1}^{n_2-1}  \tbinom{n_2-1}{\ell-1}\tbinom{n_1-1}{n_2-1}V_{\ell}\bigl(1 - F(\tilde{a}_i) \bigl)^{\ell-1} Q(a_i,\tilde{a}_i\mid n_1,n_2,n_2 - \ell)+ \tbinom{n_1-1}{n_2-1}V_{n_2} F^{n_1-n_2}(a_i)\left((1 - F(\tilde{a}_i) \right)^{n_2-1}.
\end{align*}
Then 
\begin{align*}
    \frac{\partial S(a_i,\tilde{a}_i\mid n_1,n_2)}{\partial \tilde{a}_i} \leq & \frac{\partial S(\tilde{a}_i,\tilde{a}_i\mid n_1,n_2)}{\partial \tilde{a}_i} = \sum\limits_{\ell=1}^{n_2-1}(V_\ell-V_{\ell+1})F'_{(\ell,n_1-1)}(\tilde{a}_i),
\end{align*}
where for $1\leq \ell \leq n_1$, $F_{(\ell,n)}(\cdot)$ is the distribution function of the $\ell^{\textrm{th}}$ largest order statistic among $n$ i.i.d. random variables with distribution $F(\cdot)$, and for $x \in[0,1]$,
\begin{align*}
F_{(\ell,n_1)}(x) =& \sum_{j=n_1-\ell+1}^{n_1} \binom{n_1}{j}F^j(x)\bigl(1 - F(x) \bigl)^{n_1-j},\\
F'_{(\ell,n_1)}(x) =& \frac{n_1!}{(n_1-\ell)!(\ell-1)!}F^{n_1-\ell}(x)\bigl(1 - F(x) \bigl)^{\ell-1}f(x).
\end{align*}
\end{lemma}
\begin{proof}
We start by calculating $\partial S(a_i,\tilde{a}_i\mid n_1,n_2)/\partial \tilde{a}_i$. 
\begin{align*}
    \frac{\partial S(a_i,\tilde{a}_i\mid n_1,n_2)}{\partial \tilde{a}_i} =& - \tbinom{n_1-1}{n_2-1}\sum_{\ell=2}^{n_2-1} \tbinom{n_2-1}{\ell-1}  (\ell -1) V_{\ell}\bigl((1 - F(\tilde{a}_i) \bigl)^{\ell-2}Q(a_i,\tilde{a}_i\mid n_1,n_2,n_2 - \ell) f(\tilde{a}_i)\\
    & + \tbinom{n_1-1}{n_2-1}\sum_{\ell=1}^{n_2-1}  \tbinom{n_2-1}{\ell-1}V_{\ell} \bigl(1 - F(\tilde{a}_i) \bigl)^{\ell-1} \frac{\partial Q(a_i,\tilde{a}_i\mid n_1,n_2,n_2 - \ell)}{\partial \tilde{a}_i}\\
    & - (n_2 - 1)\tbinom{n_1-1}{n_2-1}V_{n_2} F^{n_1-n_2}(a_i)\left((1 - F(\tilde{a}_i) \right)^{n_2-2}f(\tilde{a}_i).
\end{align*}
 We then show that $\partial S(a_i,\tilde{a}_i\mid n_1,n_2)/\partial \tilde{a}_i$ is increasing in $a_i \in [0, \tilde{a}_i]$. This can be checked by
\begin{align*}
    &\frac{\partial S^2(a_i,\tilde{a}_i\mid n_1,n_2)}{\partial \tilde{a}_i \partial {a}_i}\\ =& - \tbinom{n_1-1}{n_2-1}\sum_{\ell=2}^{n_2-1}  \tbinom{n_2-1}{\ell-1} (\ell -1)V_{\ell}\bigl((1 - F(\tilde{a}_i) \bigl)^{\ell-2}\frac{\partial Q(a_i,\tilde{a}_i\mid n_1,n_2,n_2 - \ell)}{\partial a_i}f(\tilde{a}_i)
    \\
    & + \tbinom{n_1-1}{n_2-1}\sum_{\ell=1}^{n_2-1}  \tbinom{n_2-1}{\ell-1}V_{\ell} \bigl(1 - F(\tilde{a}_i) \bigl)^{\ell-1} \frac{\partial Q^2(a_i,\tilde{a}_i\mid n_1,n_2,n_2 - \ell)}{\partial \tilde{a}_i \partial {a}_i}\\
    & - (n_1 - n_2)(n_2 - 1)\tbinom{n_1-1}{n_2-1}V_{n_2} F^{n_1-n_2-1}(a_i)\left((1 - F(\tilde{a}_i) \right)^{n_2-2}f(\tilde{a}_i)f({a}_i)\\
    =&  -(n_1 - n_2)\tbinom{n_1-1}{n_2-1}F^{n_1 - n_2 -1}({a}_i)\sum_{\ell=2}^{n_2}  \tbinom{n_2-1}{\ell-1} (\ell -1)V_{\ell}\bigl((1 - F(\tilde{a}_i) \bigl)^{\ell-2}(F(\tilde{a}_i) - F({a}_i))^{n_2 - \ell}f(\tilde{a}_i)f({a}_i)\\
    &+ (n_1 - n_2) \tbinom{n_1-1}{n_2-1} F^{n_1 - n_2 -1}({a}_i)\sum_{\ell=1}^{n_2-1}  \tbinom{n_2-1}{\ell-1} (n_2 - \ell)V_{\ell}\bigl(1 - F(\tilde{a}_i) \bigl)^{\ell-1} (F(\tilde{a}_i) - F({a}_i))^{n_2 - \ell -1 }f(\tilde{a}_i)f({a}_i) \\
    =&  -(n_1 - n_2)\tbinom{n_1-1}{n_2-1}F^{n_1 - n_2 -1}({a}_i)\sum_{\ell=1}^{n_2 - 1}  \tbinom{n_2-1}{\ell} \ell V_{\ell + 1} \bigl((1 - F(\tilde{a}_i) \bigl)^{\ell-1}(F(\tilde{a}_i) - F({a}_i))^{n_2 - \ell -1}f(\tilde{a}_i)f({a}_i)\\
    &+ (n_1 - n_2)  \tbinom{n_1-1}{n_2-1}F^{n_1 - n_2 -1}({a}_i)\sum_{\ell=1}^{n_2-1}  \tbinom{n_2-1}{\ell-1} (n_2 - \ell) V_{\ell}\bigl(1 - F(\tilde{a}_i) \bigl)^{\ell-1} (F(\tilde{a}_i) - F({a}_i))^{n_2 - \ell -1 }f(\tilde{a}_i)f({a}_i)\\
    =& (n_1 - n_2) \tbinom{n_1-1}{n_2-1} F^{n_1 - n_2 -1}({a}_i) \sum_{\ell=1}^{n_2-1}  \tbinom{n_2-1}{\ell-1} (n_2 - \ell)( V_{\ell}- V_{\ell+1})\bigl(1 - F(\tilde{a}_i) \bigl)^{\ell-1}(F(\tilde{a}_i) - F({a}_i))^{n_2 - \ell -1 }f(\tilde{a}_i)f({a}_i)\\
    \geq & 0,
\end{align*}
where the least equality holds by the equation $\tbinom{n_1-1}{n_2-1} (n_2 - \ell) = \tbinom{n_2-1}{\ell}\ell$.
Hence
\begin{align*}
    \frac{\partial S(a_i,\tilde{a}_i\mid n_1,n_2)}{\partial \tilde{a}_i} \leq & \frac{\partial S(\tilde{a}_i,\tilde{a}_i\mid n_1,n_2)}{\partial \tilde{a}_i}\\ 
    = & -\tbinom{n_1-1}{n_2-1}\sum_{\ell=2}^{n_2-1}  \tbinom{n_2-1}{\ell-1} (\ell -1)V_{\ell}\bigl((1 - F(\tilde{a}_i) \bigl)^{\ell-2}F'(\tilde{a}_i)Q(\tilde{a}_i,\tilde{a}_i\mid n_1,n_2,n_2 - \ell)\\
    & + \tbinom{n_1-1}{n_2-1}\sum_{\ell=1}^{n_2-1}  \tbinom{n_2-1}{\ell-1}V_{\ell} \bigl(1 - F(\tilde{a}_i) \bigl)^{\ell-1} \frac{\partial Q(\tilde{a}_i,\tilde{a}_i\mid n_1,n_2,n_2 - \ell)}{\partial \tilde{a}_i}\\
    & - (n_2 - 1)\tbinom{n_1-1}{n_2-1}V_{n_2} F^{n_1-n_2}(a_i)\left((1 - F(\tilde{a}_i) \right)^{n_2-2}f(\tilde{a}_i)\\
    = & -\tbinom{n_1-1}{n_2-1}\sum_{\ell=2}^{n_2-1}   \frac{\tbinom{n_2-1}{\ell-1}}{\tbinom{n_1-\ell}{n_2-\ell}} (\ell -1)V_{\ell}\bigl((1 - F(\tilde{a}_i) \bigl)^{\ell-2}F^{n_1 - \ell}(\tilde{a}_i) f(\tilde{a}_i) \\
    & + \tbinom{n_1-1}{n_2-1}\sum_{\ell=2}^{n_2}   \frac{\tbinom{n_2-1}{\ell-1}}{\tbinom{n_1-\ell}{n_2-\ell}} (n_1-\ell)V_{\ell}\bigl((1 - F(\tilde{a}_i) \bigl)^{\ell-1}F^{n_1 - \ell -1}(\tilde{a}_i) f(\tilde{a}_i)\\
    =& -\sum_{\ell=2}^{n_2} V_{\ell} F'_{(\ell - 1,n_1-1)} + \sum_{\ell=1}^{n_2 - 1} V_{\ell} F'_{(\ell,n_1-1)}\\
     =&\sum\limits_{\ell=1}^{n_2-1}(V_\ell-V_{\ell+1})F'_{(\ell,n_1-1)}(\tilde{a}_i),
\end{align*}
where $Q(\tilde{a}_i,\tilde{a}_i\mid n_1,n_2,n_2 - \ell)$ and $\partial Q(\tilde{a}_i,\tilde{a}_i\mid n_1,n_2,n_2 - \ell)/\partial \tilde{a}_i$ are obtained by Lemma \ref{Auxiliary Lemma_derivative}, and the penultimate equality uses the fact that $\tbinom{n_1-1}{n_2-1} \tbinom{n_2-1}{\ell-1} = \tbinom{n_1-\ell}{n_2-\ell} \tbinom{n_1-1}{\ell-1}$ and 
    $$F'_{(\ell,n_1-1)}(x) - F'_{(\ell-1,n_1-1)}(x)=\tbinom{n_1-1}{\ell-1}\left(\big(1 - F(x) \big)^{\ell-1} F^{n_1-\ell}(x)\right)'.$$
This completes the proof of the lemma.
\end{proof}

\section{A Two-Stage Model and Relevant Proofs}
\label{sec_appendix_sec}
In this section, we give the formal definition of a two-stage model and present its analysis. %

\subsection{A Two-Stage Model}

Consider a two-stage model where a set of $n_1 \in \mathbb{N}_{>0}$ ($n_1\geq 2$) players (denoted as set $\mathcal{I}^1$) enter the first-stage contest and depending on the first-stage efforts, the top $n_2\in \mathbb{N}_{>0}$ ($2\leq n_2\leq n_1$) players (denoted as set $\mathcal{I}^2$) proceed to the second stage to compete for $p\in \mathbb{N}_{>0}$ (wlog, $p = n_2$) prizes. The value of the $\ell^{\textrm{th}}$ prize is denoted by $V_\ell$, where $V_1 \geq V_2 \geq \dots \geq V_{n_2} \geq 0$.
For $i \in \mathcal{I}^t$, let $\mathcal{I}_{-i}^t = \left\{j: j\in \mathcal{I}^t, j\neq i \right\}$ as the set of players in stage $t \in \{1,2\}$ of the contest except for player $i$.

At each stage $t\in\{1,2\}$ of the contest, a player $i\in\mathcal{I}^t$ exerts an effort $e_i^t$ which incurs a cost $g(e_i^t)/a_i$ where $g: \mathbb R_{\ge0} \mapsto \mathbb R_{\ge0}$ is a strictly increasing, continuous and differentiable function with $g(0)=0$ and $a_i$ is a type/ability parameter. For notation brevity, here we assume that the first-stage cost function is the same as the second-stage one. {Our result still holds when two stages' cost functions are different.}
The ability $a_i$ is private information to player $i$ herself. At the beginning of the first stage, each player perceives other players' abilities as i.i.d. random variables drawn from a commonly known distribution function $F(\cdot)$ with continuous density $f(\cdot)$. Let $A_i$ be the random variable (perceived by other players) of player $i$'s ability. Note that the designer holds the same prior belief. Without loss of generality, we assume that $a_i\in[0,1]$.

At each stage, eligible players will simultaneously and privately choose their efforts. Players whose first-stage efforts are among the top $n_2$ of all first-stage efforts enter the second stage. The player with the highest second-stage effort wins the first prize, $V_1$; similarly, the player with the second-highest second-stage effort wins the second prize, $V_2$, and so on until all the prizes are allocated. Player $i$ with ability $a_i$ gets utility $V_\ell - (g(e_i^1) + g(e_i^2))/a_i$ if she wins the $\ell^{\textrm{th}}$ prize with efforts $e_i^1$ and $e_i^2$ in both stages; if player $i$ does not win any prize, she gets utility $-g(e_i^1)/a_i$ in the case where she gets eliminated after the first stage, or $- (g(e_i^1) + g(e_i^2))/a_i$ in the case where she enters the second stage but does not land any prize. Players are strategic and risk-neutral, they choose their efforts in both stages in order to maximize their expected utilities, where the expectations are taken over the randomness of competing players' efforts. %
We will give a formal definition of expected utilities later. %

\medskip

\noindent\textbf{Belief System.} We assume that at the end of the first stage, the contest designer will publicly announce who is advanced and who is eliminated.
For any $i \in \mathcal{I}^1$, let $s_i \in \{0,1\}$ be the \textit{signal} of player $i$, where $s_i=1$ is promoted to attend second-stage contest, and $s_i=0$ otherwise. For $i\in \mathcal{I}^1$, let $s_{-i}=[s_j:j\in \mathcal{I}^1_{-i}]$ and $s = [s_i,s_{-i}]$. 

For any admitted player $i\in \mathcal{I}^2$, let $\beta_i(a_j\mid s,e_i^1)$ denote her posterior (second-stage) belief (probability density) about admitted player $j\in \mathcal{I}_{-i}^2$. For $i\in \mathcal{I}^2$, let $a_{-i}^2=[a_j:j\in \mathcal{I}_i^2]$ be all other admitted players' abilities expected for admitted player $i$. Let $A_{-i}^2=[A_j:j\in \mathcal{I}_{-i}^2]$ be the corresponding random variables. With a slight abuse of notation, for any $i\in \mathcal{I}^2$, denote by $\beta_i(a_{-i}\mid s,e_i^1)$ the admitted player $i$'s belief (joint probability density) about all other $n_2-1$ players' abilities conditional on admission signal $s$ and her first-stage effort $e_i^1$.

\medskip

\noindent\textbf{Equilibrium.} 
Let $e^1 = [e_i^1:i \in \mathcal{I}^1]$ be all players' first-stage efforts, and $e^2=[e_i^2:i\in \mathcal{I}^2]$ be all admitted players' second-stage efforts.

\begin{definition}[Formal Definition of PBE]
\label{def_PBE in SEC}
A perfect Bayesian equilibrium (PBE) is a tuple of a strategy $[e^1,e^2]$ and posterior beliefs $\left[\beta_i(\cdot\mid s, e_i^1): i\in\mathcal{I}^2\right]$ that satisfies the following conditions.

\begin{enumerate}%
\item[(i)] \emph{Bayesian Updating:} for every player $i \in \mathcal{I}^2$, Bayes' rule is used to update her posterior belief $\beta_i(\cdot\mid s , e_i^1)$, 
\begin{align*}
    \beta_i(a_{-i}^2 \mid s ,e_i^1) = \beta_i(a_{-i}^2 \mid s_{-i},s_{i},e_i^1) = \frac{\Pr\left(s_{-i} \mid a_{-i}^2,s_{i}, e_i^1\right)\prod_{j\in\mathcal{I}^2_{-i}}f(a_j)}{\int_{a_{-i}^2}\Pr\left(s_{-i} \mid a_{-i}^2,s_{i}, e_i^1\right)\prod_{j\in\mathcal{I}^2_{-i}}f(a_j)da_{-i}^2}. %
\end{align*}

The last equation holds for a similar reason in the proof of Proposition \ref{prop_post_belief}.

\item[(ii)]  \emph{Sequential Rationality:} player $i\in\mathcal{I}^t$’s expected utility starting at any stage $t\in\{1,2\}$ of the contest is maximized by playing $e_i^t$ given her belief at stage $t$,
\begin{align*}
    e_i^2 \in \normalfont{\argmax}_{\Tilde{e}_i^2}~\sum_{\ell=1}^p P_{i,\ell}\cdot V_\ell  - \frac{g\left(\Tilde{e}_i^2\right)}{a_i},\quad
    e_i^1 \in \normalfont{\argmax}_{\Tilde{e}_i^1}~P_i\cdot u_i^2  - \frac{g\left(\Tilde{e}_i^1\right)}{a_i},
\end{align*}
where $u_i^2 := \normalfont{\max}_{\Tilde{e}_i^2}~ \sum_{\ell=1}^p P_{i,\ell}\cdot V_\ell  - g\left(\Tilde{e}_i^2\right)/a_i$. $P_{i,\ell}$ is the probability of a player $i\in \mathcal{I}^2$ winning the $l^\textrm{th}$ prize based on her second-stage belief, and $P_{i}$ is the probability of a player $i\in\mathcal{I}^1$ advancing to the second stage based on her first-stage belief. 
\end{enumerate}    
\end{definition}

\smallskip

We now give the formal definition of a symmetric and strictly increasing PBE stated in Proposition \ref{prop_non-existenc_sec}.
\begin{definition}[Symmetric and Strictly Increasing PBE]
\label{def_symmetric and strictly monotone PBE}
We call a PBE, which is a tuple of a strategy $[e^1,e^2]$ and posterior beliefs $\left[\beta_i(\cdot\mid s, e_i^1): i\in\mathcal{I}^2\right]$ satisfying Bayesian updating and sequential updating, symmetric and strictly increasing if and only if for all $i\in \mathcal{I}^t$, $e_i^t=b^t(a_i), t=1,2$, where $b^t(\cdot)$ is strictly increasing and $b^t(0)=0$.
\end{definition}

\subsection{Non-Existence Proof}
\label{subsec:non_existence}

\begin{proof}[\textbf{Proof of Proposition \ref{prop_non-existenc_sec}}]
The proof includes three steps: i) we first characterize the posterior beliefs based on monotone and symmetric first-stage strategy; ii) when $v(\cdot\mid n_2)$ is not increasing, there does not exist 
symmetric and strictly increasing PBE;
iii) when $v(\cdot\mid n_2)$ is increasing, we derive the unique symmetric and strictly increasing PBE if existing, but we show that there is always an incentive for players to deviate from the above strategy. 

\medskip

\noindent\textit{\underline{Step 1: Posterior beliefs based on symmetric and strictly increasing first-stage strategy}}

\smallskip

\begin{lemma}
\label{lemma_sec_posterior}
Assuming that all players in the first stage follow the symmetric and strictly increasing strategy $b^1(\cdot)$, i.e., $e^1_i = b^1(a_i),\forall i \in \mathcal{I}^1$, player $i$'s ($i\in \mathcal{I}^2$) posterior belief in the equilibrium path is
\begin{align*}
\beta_i(a_{-i}^2 \mid s ,e_i^1) = 
\begin{cases}
\frac{F^{n_1-n_2}(a_i)}{I(F(a_i),n_2)}\prod\limits_{j \in \mathcal{I}_{-i}^2}f(a_j), \qquad & a_i \leq \normalfont{\min}_{j\in \mathcal{I}_{-i}^2} a_j , \\[4mm]
\frac{F^{n_1-n_2}\left(\normalfont{\min}_{j\in \mathcal{I}_{-i}^2} a_j\right)}{I(F(a_i),n_2)}\prod\limits_{j \in \mathcal{I}_{-i}^2}f(a_j), \quad \qquad & a_i > \normalfont{\min}_{j\in \mathcal{I}_{-i}^2} a_j.
\end{cases}
\end{align*}
\end{lemma}

\begin{proof}[Proof]
Assuming that all players adopt the symmetric and strictly increasing strategy in the first stage, in the equilibrium path, the ranking of first-stage efforts is exactly the same as the ranking of all players' abilities. In this case, the posteriors would be the same as the posterior beliefs characterized in Proposition \ref{prop_post_belief}.
\end{proof}

\smallskip

\noindent\textit{\underline{Step 2: There does not exist 
symmetric and strictly increasing PBE when $v(\cdot\mid n_2)$ is not increasing}}

\smallskip

\noindent If the posterior beliefs are given in Lemma \ref{lemma_sec_posterior}, then due to the same reason as in Proposition \ref{thm_equilibrium strategy_general_prize}, there does not exist a strictly increasing second-stage equilibrium strategy, which implies that there does not exist symmetric and strictly increasing two-stage PBE.

\medskip

\noindent\textit{\underline{Step 3: There does not exist 
symmetric and strictly increasing PBE when $v(\cdot\mid n_2)$ is increasing}}

\medskip

\noindent\textit{\underline{Part (a): symmetric and strictly increasing PBE if exists}}

\begin{lemma}
\label{lemma_sec_eq}
If there exists symmetric and strictly increasing PBE, the posterior beliefs are given in Lemma \ref{lemma_sec_posterior}, and $v(\cdot\mid n_2)$ is increasing, then the unique equilibrium strategies are given as follows. For any player $i\in \mathcal{I}^2$, her second-stage effort is
\begin{align*}
   \label{Eq_2nd Equilibrium Strategy}
    b^2(a_i) =g^{-1}\left(   \int_{0}^{a_i}v(x\mid n_2)\sum\limits_{\ell=1}^{n_2-1}(V_\ell-V_{\ell+1})dF_{(\ell,n_1-1)}(x)\right).
\end{align*}
For any player $i\in \mathcal{I}^1$, her first-stage effort is
\begin{equation*}
   b^1(a_i) = g^{-1}\biggl(u_i^2  \int_0^{a_i}x dF_{(n_2,n_1-1)}(x)  \biggl),
\end{equation*}
where $u_i^2$ is the expected utility of player $i$ starting at the second stage under strategy $b^2(a_i)$,
\[
u_i^2 =\int_{0}^{a_i}\sum\limits_{\ell=1}^{n_2-1}(V_\ell-V_{\ell+1})F_{(\ell,n_1-1)}(x)dv(x\mid n_2).\]
\end{lemma}

\begin{proof}[Proof]
Assuming that $v(\cdot\mid n_2)$ is increasing, and there exists a symmetric and strictly increasing PBE, we derive $b^2(\cdot)$ and $b^1(\cdot)$ by backward induction.  This is doable since if all players follow the symmetric and strictly increasing strategy in the first stage, the ranking of the first-stage efforts is exactly the ranking of players' abilities, i.e., the posterior beliefs are independent of players' first-stage efforts. In this case, the second-stage equilibrium is exactly the same as the equilibrium under a one-round elimination contest characterized in Proposition \ref{thm_equilibrium strategy_general_prize}.

In the first stage, every player $i \in \mathcal{I}^1$  maximizes the following objective function based on her first-stage belief,
\[
    \max_{e_i^1}~ P_i\left(\sum_{\ell=1}^{n_2} V_\ell P_{i,\ell} - \frac{g(b^2(a_i))}{a_i}\right) - \frac{g(e_i^1)}{a_i}  = \max_{e_i^1} ~  P_i u_i^2 - \frac{g(e_i^1)}{a_i},
\]
where 
\begin{align*}
    u_i^2 &=\sum\limits_{\ell=1}^{n_2} \frac{V_{\ell}}{J(F(a_i),n_2)} \tbinom{n_1-1}{\ell-1}\bigl(1 - F(a_i)\bigl)^{\ell-1}F^{n_1-\ell}(a_i) - \frac{g(b^2(a_i))}{a_i}\\
    & = \frac{1}{J(F(a_i),n_2)}\sum\limits_{\ell=1}^{n_2-1}(V_\ell-V_{\ell+1})F_{(\ell,n_1-1)} - \frac{1}{a_i} \int_{0}^{a_i}v(x\mid n_2)\sum\limits_{\ell=1}^{n_2-1}(V_\ell-V_{\ell+1})dF_{(\ell,n_1-1)}(x)\\
    & = \int_{0}^{a_i}\sum\limits_{\ell=1}^{n_2-1}(V_\ell-V_{\ell+1})F_{(\ell,n_1-1)}(x)dv(x\mid n_2),
\end{align*}
and 
\begin{align}
    P_i & = \sum_{\ell=1}^{n_2} \Pr\Big(\text{$e_i^1$ ranks $\ell^\textrm{th}$ highest in $\{b^1(A_j): j\in\mathcal{I}_{-i}^1\}\cup\{e_i^1\}$}\Big) \nonumber \\
    &= \sum_{\ell=1}^{n_2} \binom{n_1-1}{\ell-1} F^{n_1-\ell}(\gamma_1(e_i^1)) \bigl(1 - F(\gamma_1(e_i^1)) \bigl)^{\ell-1},\nonumber\\
    & = \sum_{\ell=1}^{n_2}\bigl(F_{(\ell,n_1-1)}(\gamma_1(e_i^1)) - F_{(\ell-1,n_1-1)}(\gamma_1(e_i^1)\bigl)  = {F}_{(n_2,n_1-1)}(\gamma_1(e_i^1)), \nonumber
\end{align}
where $\gamma_1(\cdot)$ is the inverse function of $b^1(\cdot)$.

Since $u_i^2$ is independent of $e_i^1$, following the same logic as deriving the second-stage equilibrium strategy (first derive the first-order condition, then solve the differential equation), we get the unique first-stage equilibrium effort as follows.
\begin{equation*}
\label{Eq_Proof_First-stage Equilibrium}
     b^1(a_i) = g^{-1}\biggl(u_i^2  \int_0^{a_i}x dF_{(n_2,n_1-1)}(x)  \biggl).
\end{equation*}   
This completes the proof of the lemma.
\end{proof}

\smallskip

\noindent\textit{\underline{Part (b): Players always have the incentive to deviate}}

\smallskip

We now show that there is always incentive for player $i$ to deviate from {$\left\{b^1(\cdot),b^2(\cdot)\right\}$} given all other players follow the strategy {$\left\{b^1(\cdot),b^2(\cdot)\right\}$}. Firstly, given all players follow $b^1(\cdot)$ in the first stage, no admitted player has the incentive to deviate from $b^2(\cdot)$ in the second stage. However, we will show that players do have the incentive to deviate from $b^1(\cdot)$. 

Given all other players follow {$\left\{b^1(\cdot),b^2(\cdot)\right\}$}, suppose that (admitted) player $i$ in the first-stage deviates from $b^1(a_i)$ to $\Tilde{e}_i^1\neq b^1(a_i)$. Since $b^1(\cdot)$ is increasing, we assume that $\Tilde{e}_i^1 = b^1(\Tilde{a}_i)$, $\Tilde{a}_i\neq a_i$. Given all other players follow {$\left\{b^1(\cdot),b^2(\cdot)\right\}$}, i) we first derive the  posterior beliefs of player $i\in \mathcal{I}^2$ based on her deviated first-stage efforts $b^1(\Tilde{a}_i)$; ii) we then use backward induction to derive player $i$'s optimal second-stage strategy $\Tilde{b}^2(a_i)$; iii) finally we show that $a_i \notin \argmax_{\Tilde{a}_i} u_i(b^1(\Tilde{a}_i), \Tilde{b}^2(a_i))$, where $u_i(b^1(\Tilde{a}_i), \Tilde{b}^2(a_i))$ is player $i$'s expected utility given player $i$ follows {$\left\{b^1(\Tilde{a}_i), \Tilde{b}^2(a_i)\right\}$} and all other players follow {$\left\{b^1(\cdot),b^2(\cdot)\right\}$}. 

\smallskip

Denote $\gamma_2(\cdot)$ as the inverse function of $b^2(\cdot)$.  
Given all players follow $b^1(\cdot)$ in the first stage, if player $i\in \mathcal{I}^2$ with ability $a_i$ deviates to $b^1(\Tilde{a}_i)$, her posterior beliefs become
\begin{align*}
    \beta_i(a_{-i}^2 \mid s ,b^1(\Tilde{a}_i)) = \frac{F^{n_1-n_2}\left(\normalfont{\min}_{j\in \mathcal{I}_{-i}^2} a_j \wedge \tilde{a}_i\right)}{I(F(\tilde{a}_i)\mid n_2)}\prod\limits_{j \in \mathcal{I}_{-i}^2}f(a_j)
\end{align*}

We now derive player $i$'s optimal second-stage strategy $\Tilde{b}^2(a_i)$ based on player $i$'s deviated first-stage effort $b^1(\Tilde{a}_i)$. Let us consider player $i$'s second-stage utility function
\begin{equation*}
	\max_{e_i^2}~ \sum_{\ell=1}^{n_2} V_\ell P_{i,\ell}(\gamma_2(e_i),\tilde{a}_i) - \frac{g(e_i^2)}{a_i},
\end{equation*}
where    
\begin{align}
		P_{i,\ell}(\gamma_2(e_i),\tilde{a}_i) & = \Pr\Big(\textrm{$e_i^2$ ranks $\ell^\textrm{th}$ highest in $\left\{b^2(A_j): j\in\mathcal{I}_{-i}^2\right\}\cup\left\{e_i^2\right\}$}\Big) \nonumber \\[2mm]
		& = \tbinom{n_2-1}{\ell-1} \underbrace{\int_{\gamma_2(e_{i}^2)}^1\cdots \int_{\gamma_2(e_{i}^2)}^1}_{\ell-1} \underbrace{\int_0^{\gamma_2(e_{i}^2)}\cdots\int_0^{\gamma_2(e_{i}^2)}}_{n_2-\ell} \beta_i(a_{-i}^2 \mid s ,b^1(\Tilde{a}_i))da_{-i}^2 \nonumber.
	\end{align} Putting $P_{i,\ell}(\gamma_2(e_i),\tilde{a}_i)$ back, player $i$'s second-stage utility function is given by    
     \begin{equation}\label{non existence 0}
     	\max_{e_i^2} ~\sum_{\ell=1}^{n_2} \tbinom{n_2-1}{\ell-1} V_\ell\underbrace{\int_{\gamma_2(e_{i}^2)}^1\cdots \int_{\gamma_2(e_{i}^2)}^1}_{\ell-1} \underbrace{\int_0^{\gamma_2(e_{i}^2)}\cdots\int_0^{\gamma_2(e_{i}^2)}}_{n_2-\ell} \beta_i(a_{-i}^2 \mid s ,b^1(\Tilde{a}_i))da_{-i}^2 - \frac{g(e_i^2)}{a_{i}}. 
     \end{equation}
Our ultimate goal is to derive $u_i(b^1(\Tilde{a}_i), \Tilde{b}^2(a_i))$, thus we don't need to derive the closed-form solution of \eqref{non existence 0}. With a little abuse of notation, denote $\tilde{b}^2(a_i,\tilde{a}_i)$ as the optimal solution of the above optimization problem. We write it as a function of $a_i$ and $\Tilde{a}_i$ to emphasize its dependency on deviated ability $\Tilde{a}_i$. Define $\eta=\gamma_2\circ \tilde{b}^2$ as the composition of function $\gamma_2(\cdot)$ and function $\tilde{b}_2(\cdot)$. Since when $\Tilde{a}_i=a_i$, $\tilde{b}^2(a_i,\tilde{a}_i)  =b^2(a_i)$, we have $\eta(a_i,\tilde{a}_i) = a_i$.
Recall that, 
\[
b^2(a_i) =g^{-1}\left(   \int_{0}^{a_i}v(x\mid n_2)\sum\limits_{\ell=1}^{n_2-1}(V_\ell-V_{\ell+1})dF_{(\ell,n_1-1)}(x)\right),
\]
hence 
\begin{align}
    \label{non existence 2}
 g\left(\tilde{b}^2(a_i,\tilde{a}_i)\right)  = g(b^2\circ \gamma_2 \circ \tilde{b}^2(a_i,\tilde{a}_i)) = g(b^2(\eta(a_{i},\tilde{a}_i))) =   \int_0^{\eta(a_{i},\tilde{a}_i)} v(x\mid n_2)\sum\limits_{\ell=1}^{n_2-1}(V_\ell-V_{\ell+1})dF_{(\ell,n_1-1)}(x).
\end{align}
Define $\tilde{u}_i^2(\tilde{b}^2(a_i,\tilde{a}_i))$ as the expected maximum second-stage utility of player $i$ given her deviated first-stage effort $b^1(\tilde{a}_i)$. Formally, we have
\begin{align*}
\tilde{u}_i^2(\tilde{b}^2(a_i,\tilde{a}_i)) =\max_{e_i^2} ~\sum_{\ell=1}^{n_2} P_{i,\ell}(\gamma_2(e_i),\tilde{a}_i) - \frac{g(e_i^2)}{a_{i}} 
=\sum_{\ell=1}^{n_2} 
P_{i,\ell}(\gamma_2\circ\tilde{b}^2(a_i,\tilde{a}_i),\tilde{a}_i)- \frac{g\left(\tilde{b}^2(a_i,\tilde{a}_i)\right)}{a_{i}}.
\end{align*}
By definition, we have $\gamma_2\left(\tilde{b}^2(a_i,\tilde{a}_i)\right) = \eta(a_i,\tilde{a}_i)$, and replace $g\left(\tilde{b}^2(a_i,\tilde{a}_i)\right)$ by equation \eqref{non existence 2}, 
\begin{align*}
\tilde{u}_i^2(\tilde{b}^2(a_i,\tilde{a}_i)) =&
\sum_{\ell=1}^{n_2} V_\ell
P_{i,\ell}(\eta(a_{i},\tilde{a}_i),\tilde{a}_i)
- \frac{1}{a_{i}} \int_0^{\eta(a_{i},\tilde{a}_i)} v(x\mid n_2)\sum\limits_{\ell=1}^{n_2-1}(V_\ell-V_{\ell+1})dF_{(\ell,n_1-1)}(x).
\end{align*}
We now have player $i$'s expected (overall) utility based on deviated first-stage effort $b^1(\tilde{a}_i)$ and optimal second-stage strategy $\tilde{b}^2(a_i,\tilde{a}_i)$ as follows.
\begin{align*}
  u_i\left(b^1(\Tilde{a}_i), \Tilde{b}^2(a_i, \tilde{a}_i)\right) = P_i(b^1(\tilde{a}_i))\cdot \tilde{u}_i^2(\tilde{b}^2(a_i,\tilde{a}_i)) - \frac{g(b^1(\tilde{a}_i))}{a_i},
\end{align*}
where $P_i(b^1(\tilde{a}_i))$ is player $i$'s promotion probability based on her first-stage effort $b^1(\tilde{a}_i)$ (given other players still follow the equilibrium strategy $b^1(\cdot)$). Thus, we have $P_i(b^1(\tilde{a}_i)) = {F}_{(n_2,n_1-1)}(\tilde{a}_i)$.
Since
\[
b^1(\tilde{a}_i)  = g^{-1}\biggl(u_i^2  \int_0^{\tilde{a}_i}x dF_{(n_2,n_1-1)}(x)  \biggl),
\]
where
\begin{align*}
u_i^2&=\sum\limits_{\ell=1}^{n_2} V_\ell P_{i,\ell}(\gamma_2(b^2(a_i)),a_i) - \frac{g(b^2(a_i))}{a_i}=\int_{0}^{a_i}\sum\limits_{\ell=1}^{n_2-1}(V_\ell-V_{\ell+1})F_{(\ell,n_1-1)}(x)dv(x\mid n_2).\nonumber
\end{align*} 
Putting all of them back to the formula of $ u_i\left(b^1(\Tilde{a}_i), \Tilde{b}^2(a_i,\tilde{a}_i)\right)$, we have
\begin{align*}
L(\tilde{a}_i, a_i):=&u_i\left(b^1(\Tilde{a}_i), \Tilde{b}^2(a_i,\tilde{a}_i)\right)\\
\quad 	=&\left(\sum_{\ell=1}^{n_2} V_\ell
P_{i,\ell}(\eta(a_{i},\tilde{a}_i),\tilde{a}_i)
- \frac{1}{a_{i}} \int_0^{\eta(a_{i},\tilde{a}_i)} v(x\mid n_2)\sum\limits_{\ell=1}^{n_2-1}(V_\ell-V_{\ell+1})dF_{(\ell,n_1-1)}(x) \right) {F}_{(n_2,n_1-1)}(\tilde{a}_i)\\ \nonumber
	&-\frac{1}{a_i} \left(\int_{0}^{a_i}\sum\limits_{\ell=1}^{n_2-1}(V_\ell-V_{\ell+1})F_{(\ell,n_1-1)}(x)dv(x\mid n_2)\right)\int_0^{\tilde{a}_i} x d{F}^{\prime}_{(n_2,n_1-1)}(x).
\end{align*}  
Showing $L(\tilde{a}_i, a_i)/\partial \tilde{a}_i\neq 0$ at $\tilde{a}_i=a_i$ is sufficient to show $a_i \notin \argmax_{\tilde{a}_i} u_i\left(b^1(\Tilde{a}_i), \Tilde{b}^2(a_i,\tilde{a}_i)\right)$.
\begin{align*}
	&\quad \frac{\partial L(\tilde{a}_i, a_i)}{\partial \tilde{a}_i}
	\nonumber \\&=\left( \sum_{\ell=1}^{n_2} V_\ell \frac{\partial P_{i,\ell}(\eta(a_{i},\tilde{a}_i),\tilde{a}_i)}{\partial \eta(a_{i},\tilde{a}_i)} \frac{\partial}{\partial \tilde{a}_i}\eta(a_i,\tilde{a}_i) + \sum_{\ell=1}^{n_2} V_\ell \frac{\partial P_{i,\ell}(\eta(a_{i},\tilde{a}_i),y)}{\partial y}\bigg|_{y=\tilde{a}_i} \right)  {F}_{(n_2,n_1-1)}(\tilde{a}_i)\\ \nonumber
 &\quad -\frac{1}{a_{i}} \frac{\eta(a_{i},\tilde{a}_i)}{J(F(\eta(a_{i},\tilde{a}_i))\mid n_2)}\sum\limits_{\ell=1}^{n_2-1} (V_\ell - V_{\ell + 1} ){F}^{\prime}_{(\ell,n_1-1)}(\eta(a_{i},\tilde{a}_i))\frac{\partial}{\partial \tilde{a}_i}\eta(a_i,\tilde{a}_i){F}_{(n_2,n_1-1)}(\tilde{a}_i)\\
 &\quad+\left(\sum_{\ell=1}^{n_2} V_\ell
P_{i,\ell}(\eta(a_{i},\tilde{a}_i),\tilde{a}_i)
- \frac{1}{a_{i}} \int_0^{\eta(a_{i},\tilde{a}_i)} v(x\mid n_2)\sum\limits_{\ell=1}^{n_2-1}(V_\ell-V_{\ell+1})dF_{(\ell,n_1-1)}(x) \right)  {F}'_{(n_2,n_1-1)}(\tilde{a}_i) \\
 &\quad-\frac{\tilde{a}_i}{a_i} \left(\int_{0}^{a_i}\sum\limits_{\ell=1}^{n_2-1}(V_\ell-V_{\ell+1})F_{(\ell,n_1-1)}(x)dv(x\mid n_2)\right){F}^{\prime}_{(n_2,n_1-1)}(\tilde{a}_i).
\end{align*} 
When $\tilde{a}_i=a_i$, $\eta(a_{i},\tilde{a}_i)= a_i$, by the similar logic in \Cref{thm_equilibrium strategy_general_prize}, 
\begin{align*}
    \left(\frac{\partial}{\partial \eta(a_{i},\tilde{a}_i)} P_{i,\ell}(\eta(a_{i},\tilde{a}_i),\tilde{a}_i)\right)\bigg|_{\left(\eta(a_{i},\tilde{a}_i)=a_i,\tilde{a}_i=a_i\right)}=\frac{{F}'_{(\ell,n_1-1)}(a_i) - {F}'_{(\ell - 1,n_1-1)}(a_i)}{J(F(a_i)\mid n_2)},
\end{align*}
then it can be shown that
\begin{align*}
    \frac{\partial L(\tilde{a}_i, a_i)}{\partial \tilde{a}_i}\bigg|_{\tilde{a}_i = a_i}&=\left(\sum_{\ell=1}^{n_2} V_\ell \frac{\partial P_{i,\ell}(a_i,y)}{\partial y}\bigg|_{y={a}_i} \right)  {F}_{(n_2,n_1-1)}({a}_i)\\
    &=- \frac{J'(F(a_i))f(a_i)}{J^2(F(a_i))}\left(\sum_{\ell=1}^{n_2} V_\ell ({F}_{(\ell,n_1-1)} - {F}_{(\ell - 1,n_1-1)})(a_i)\right) {F}_{(n_2,n_1-1)}({a}_i)\\
    & = -\frac{J'(F(a_i))f(a_i)}{J^2(F(a_i))}\left(\sum_{\ell=1}^{n_2} (V_\ell - V_{\ell + 1}){F}_{(\ell,n_1-1)}\right){F}_{(n_2,n_1-1)}({a}_i)<0.
\end{align*}
Hence $a_i \notin \argmax_{\tilde{a}_i} u_i\left(b^1(\Tilde{a}_i), \Tilde{b}^2(a_i,\tilde{a}_i)\right)$. 

Steps 1 to 3 complete the proof. Our proof is also valid when the first-stage cost function is different from the second-stage cost function. 
\end{proof}

}

\end{document}